\documentclass[aps,pra,11pt,onecolumn,nofootinbib,tightenlines,superscriptaddress]{revtex4-2}

\usepackage[T1]{fontenc}
\usepackage{lmodern}
\usepackage{microtype}
\usepackage{amsmath,amssymb,mathtools,amsthm}
\usepackage{braket}
\usepackage{array,booktabs}
\usepackage[table]{xcolor}
\usepackage[most,breakable]{tcolorbox}
\usepackage{placeins}

\definecolor{verylightgray}{gray}{0.96}
\definecolor{morelightgray}{gray}{0.975}
\definecolor{arimotoaccent}{HTML}{C77800}
\definecolor{mainresultaccent}{HTML}{176B87}
\definecolor{technicalaccent}{HTML}{657A45}

\usepackage{tikz}
\usetikzlibrary{positioning, arrows.meta, decorations.pathmorphing, fit, calc}

\usepackage[colorlinks=true,urlcolor=blue,citecolor=blue,linkcolor=blue]{hyperref}

\newtheorem{theorem}{Theorem}
\newtheorem{lemma}[theorem]{Lemma}
\newtheorem{proposition}[theorem]{Proposition}
\newtheorem{arimotoproposition}[theorem]{Proposition}

\newtheorem{fact}{Fact}
\newtheorem{remark}[theorem]{Remark}
\let\oldremark\remark
\renewcommand{\remark}{\oldremark\upshape}

\tcbset{
 resultbox/.style={
   enhanced,
   breakable,
   boxrule=0.5pt,
   arc=1.5pt,
   outer arc=1.5pt,
   left=5.75pt,
   right=5.75pt,
   top=4pt,
   bottom=4pt,
   grow sidewards by=10pt,
   before skip=8pt,
   after skip=8pt
 }
}
\tcolorboxenvironment{arimotoproposition}{
 resultbox,
 colback=arimotoaccent!11!white,
 colframe=arimotoaccent!38!white
}
\tcolorboxenvironment{theorem}{
 resultbox,
 colback=mainresultaccent!9!white,
 colframe=mainresultaccent!38!white
}
\tcolorboxenvironment{lemma}{
 resultbox,
 colback=black!2!white,
 colframe=technicalaccent!32!white
}
\tcolorboxenvironment{proposition}{
 resultbox,
 colback=black!2!white,
 colframe=technicalaccent!32!white
}

\newcommand{\norm}[1]{\left\lVert#1\right\rVert}
\DeclareMathOperator{\Tr}{Tr}
\DeclareMathOperator{\supp}{supp}
\DeclareMathOperator{\id}{id}

\newcommand{\ClassicalCapacity}{C}
\newcommand{\QuantumCapacity}{Q}
\newcommand{\HolevoQuantity}{C^{(1)}}
\newcommand{\CoherentInformation}{Q^{(1)}}
\newcommand{\RenyiClassicalCapacity}[1]{C_{#1}}
\newcommand{\RenyiQuantumCapacity}[1]{Q_{#1}}
\newcommand{\RenyiHolevoQuantity}[1]{C_{#1}^{(1)}}
\newcommand{\RenyiCoherentInformation}[1]{Q_{#1}^{(1)}}
\newcommand{\RenyiMutualInformation}[1]{I_{#1}}
\newcommand{\HolevoAlt}{\chi}

\begin{document}

\title{No information transmission through quantum channels above capacity}
\hypersetup{
 pdftitle={No information transmission through quantum channels above capacity},
 pdfauthor={Hao-Chung Cheng, Marco Tomamichel},
 pdfkeywords={quantum capacity, classical capacity, entanglement generation, coherent information, Holevo information, sandwiched Renyi divergence, strong converse}
}

\begin{abstract}
We show that the capacity of a quantum channel demarcates a phase transition:
while reliable transmission below capacity is always possible, any attempt to
transmit information above it fails catastrophically. Specifically,
we prove exponential strong converse theorems for unassisted quantum and classical
communication over arbitrary finite-dimensional memoryless quantum channels.
At rates beyond the respective capacity, the entanglement-generation fidelity
and the success probability for classical communication decay exponentially
with the number of channel uses. This rules out transmission above capacity
even when one tolerates arbitrarily large errors. Our proof
follows the classical Arimoto strategy, augmented by a crucial new ingredient: integral representations of Rényi information measures that lead to
asymptotic continuity bounds for R\'enyi capacities. 

\end{abstract}

\author{Hao-Chung Cheng}
\affiliation{Department of Electrical Engineering, National Taiwan University, Taipei 10617, Taiwan}
\affiliation{Physics/Mathematics Division, National Center for Theoretical Sciences, Taiwan}
\affiliation{Hon Hai (Foxconn) Quantum Computing Center, Taiwan}

\author{Marco Tomamichel}
\affiliation{Department of Electrical and Computer Engineering,
National University of Singapore, Singapore 117583, Singapore}
\affiliation{Centre for Quantum Technologies, National University of Singapore, Singapore 117543, Singapore}

\date{September 9, 2026}

\maketitle

\begingroup
\footnotesize
\textbf{Use of artificial intelligence:}
OpenAI Codex with ChatGPT~6 Astra proposed showing asymptotic continuity via derivative bounds (the crucial new ingredient) and presented a first correct but arguably opaque proof. The current
exposition of the argument via an integral representation is due to the authors. Artificial intelligence was also used to prepare the initial manuscript and helped with literature searches and revisions. The authors have reviewed the manuscript carefully and take responsibility for its contents and any remaining errors.
\par
\endgroup

\section{Introduction}
\label{sec:introduction}

The modern theory of communication began with Shannon's coding theorem for a
classical noisy channel~\cite{shannon1948mathematical}. It identifies the
channel capacity $C$ as the largest communication rate at which the decoding error can vanish
as the block length grows. Shannon's converse was weak: at rates above $C$ it rules out
decoding errors tending to zero, but it does not force the error to approach
one~\cite{fano1961transmission}.
Wolfowitz subsequently proved the strong converse, according to which the
error must instead tend to one~\cite{wolfowitz1957coding}. Gallager quantified
the exponential decrease of the error below capacity
\cite{gallager1968information}, and Arimoto gave the complementary exponential
bound on the success probability above capacity
\cite{arimoto1973converse}. These results yield the sharp transition shown
in Fig.~\ref{fig:error-rate}. A weak converse leaves open a region above
capacity in which communication could retain a fixed probability of success;
only the strong converse makes capacity a genuine phase-transition point
between asymptotically reliable communication and complete failure.

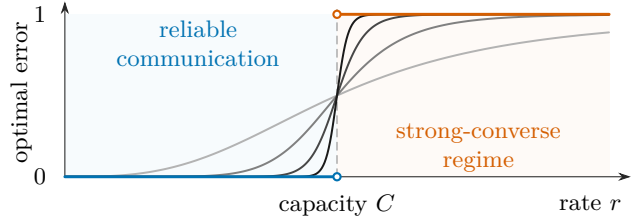
\begin{figure}[ht!]
\centering
\begingroup
\definecolor{errorrateblue}{HTML}{0072B2}
\definecolor{errorrateorange}{HTML}{D55E00}
\begin{tikzpicture}[
  x=3.6cm,y=2.15cm,font=\footnotesize,
  line cap=round,line join=round,
  >={Stealth[length=1.7mm,width=1.1mm]}
]
  \fill[errorrateblue!3] (0,0) rectangle (1,1);
  \fill[errorrateorange!3] (1,0) rectangle (2,1);
  \draw[black!45,densely dashed] (1,0) -- (1,1.03);
  \draw[->,black!85] (0,0) -- (2.08,0);
  \draw[->,black!85] (0,0) -- (0,1.08);
  \begin{scope}
    \clip (0,0) rectangle (2,1);
    \foreach \steepness/\shade in {3/30,8/50,20/72,50/90}{
      \draw[black!\shade,line width=.75pt]
        plot[domain=0:1,samples=101,variable=\x,smooth]
        (\x,{pow(\x,\steepness)/(1+pow(\x,\steepness))});
      \draw[black!\shade,line width=.75pt]
        plot[domain=1:2,samples=101,variable=\x,smooth]
        (\x,{1/(1+pow(\x,-\steepness))});
    }
  \end{scope}
  \draw[errorrateblue,line width=1.1pt] (0,0) -- (1,0);
  \draw[errorrateorange,line width=1.1pt] (1,1) -- (2,1);
  \filldraw[fill=white,draw=errorrateblue,line width=.8pt]
    (1,0) circle[radius=1.35pt];
  \filldraw[fill=white,draw=errorrateorange,line width=.8pt]
    (1,1) circle[radius=1.35pt];
  \node[anchor=east] at (-.035,0) {$0$};
  \node[anchor=east] at (-.035,1) {$1$};
  \node[rotate=90] at (-.16,.52) {optimal error};
  \node[anchor=north] at (1,-.055) {capacity $\ClassicalCapacity$};
  \node[anchor=north east] at (2.08,-.055) {rate $r$};
  \node[text=errorrateblue,align=center] at (.49,.82)
    {reliable\\communication};
  \node[text=errorrateorange,align=center] at (1.52,.18)
    {strong-converse\\regime};
\end{tikzpicture}
\endgroup
\caption{The asymptotic coding transition. Darker curves represent longer
codes. Direct coding theorems drive the optimal error to zero below capacity,
whereas a strong converse drives it to one above capacity. The corresponding
exponents quantify the approach to the limiting step. No behavior at the
transition point itself is asserted.}
\label{fig:error-rate}
\end{figure}

The goal to establish the same sharp coding transition for finite-dimensional
memoryless quantum channels $\mathcal N:A\to B$ has been elusive so far. Such a channel can carry
either classical or quantum information, and even the underlying capacity
formulas have a more intricate structure. 

For unassisted classical communication, the Holevo--Schumacher--Westmoreland coding theorem~\cite{holevo1998capacity,schumacher1997sending} 
shows that reliable communication is possible below the classical capacity, given by
\begin{equation}
 \ClassicalCapacity(\mathcal N)
 =\lim_{n\to\infty}\frac1n\,\HolevoQuantity(\mathcal N^{\otimes n}),
 \qquad
 \HolevoQuantity(\mathcal M)
 =\max_{\{p_x,\rho_A^x\}_x} I(X \!: \! B) .
 \label{eq:intro-classical-capacity}
\end{equation}
Here the optimization is over finite input ensembles $\{p_x,\rho_A^x\}_x$ and the mutual information is evaluated for the 
classical-quantum state $\omega_{XB} = \sum_xp_x\ket{x}\!\bra{x}\otimes
\mathcal M(\rho_A^x)$. Holevo had already established the corresponding weak converse~\cite{holevo1973bounds}.
The single-letter formula
$\HolevoQuantity(\mathcal M)$ is the Holevo quantity, also commonly denoted by
$\HolevoAlt(\mathcal M)$. However, the capacity is a regularized formula: the Holevo quantity is evaluated on tensor powers
of the channel, allowing the states in the ensemble to be entangled across
$n$ input copies, and the resulting quantity is normalized per copy before
taking the limit.\footnote{Thus $n$ is a regularization parameter, distinct from the
block length of an operational communication code.} This regularization is necessary and cannot in general
be replaced by their $n=1$ terms. Shor related additivity of the
Holevo quantity to other central additivity conjectures, and Hastings then
disproved them by exhibiting a strict advantage from entangled inputs
\cite{shor2004equivalence,hastings2009superadditivity}.

Moving to quantum communication, the Lloyd--Shor--Devetak coding theorem~\cite{lloyd1997capacity,shor2002capacity,devetak2005capacity} shows that quantum communication is possible below the unassisted quantum capacity, given by
\begin{equation}
 \QuantumCapacity(\mathcal N)
 =\lim_{n\to\infty}\frac1n\,\CoherentInformation(\mathcal N^{\otimes n}),
 \qquad
 \CoherentInformation(\mathcal M)
 =\max_{\rho}
    \left\{H\!\left(\mathcal M(\rho)\right)
      -H\!\left(\mathcal M^{\rm c}(\rho)\right)\right\}.
 \label{eq:intro-quantum-capacity}
\end{equation}
Here $\mathcal M^{\rm c}$ is a complementary channel. The corresponding weak converse was established in
\cite{schumachernielsen1996quantum,barnum1998information}. The single-letter formula
$\CoherentInformation(\mathcal M)$ is the coherent information.
Equation~\eqref{eq:intro-quantum-capacity}
is again regularized, with the input state optimized jointly across
$n$ channel copies.  Coherent information is
also superadditive~\cite{divincenzo1998noisy}; even more, for every fixed $n$
there are channels with $\CoherentInformation(\mathcal N^{\otimes n})=0$ but $\QuantumCapacity(\mathcal N)>0$.~\cite{cubitt2015unbounded}. Regularization is thus necessary here too.

Exponential decay of the error below capacity in Fig.~\ref{fig:error-rate} is already
established by earlier coding bounds. For classical communication,
error-exponent bounds are available for classical-quantum channels
\cite{hayashi2007error,cheng2023simple,renes2025tight,liyang2025reliability,chengliu2025packing}.
For quantum communication, quantum coding and decoupling bounds likewise yield
exponential error decay
\cite{hamada2003fidelity,morganwinter2014pretty,cheng2026joint,berta2026tight}.
These bounds give exponential decay at every rate below the respective
regularized capacities. Completing the picture in Fig.~\ref{fig:error-rate} thus only requires an exponential
strong converse above capacity.

Exponential strong converses were previously known for entanglement-assisted
communication~\cite{guptawilde2015multiplicativity,li2023strong},
classical-quantum channels
\cite{ogawa1999strong,winter1999coding,mosonyi2017strong,chenggao2024strong},
and classical communication over certain covariant
channels~\cite{koenigwehner2009strong} as well as entanglement-breaking and
Hadamard channels~\cite{Wilde3}. For quantum communication, the Rains-information bound
gives a full strong converse for generalized dephasing
channels~\cite{tomamichel2017strong}, while degradable channels initially
admitted a ``pretty strong'' converse~\cite{morganwinter2014pretty}. Full
exponential strong converses for degradable and antidegradable channels were
obtained recently~\cite{kondra2026sharp}. A recent result for stabilizer codes
over Pauli channels also treats a regularized capacity, but only within that
restricted class of codes~\cite{tomamichel2026stabilizer}. For unrestricted
unassisted communication over a general quantum channel, the regularization
and the correlations across channel uses appear to constitute the
obstacle for an exponential strong converse.

The quantum Arimoto approach bounds the fidelity or decoding success
probability through a sandwiched R\'enyi information quantity of order greater
than one~\cite{sharmawarsi2013fundamental,Wilde3,wildewinter2014erasure}.
To obtain exponential decay at a fixed rate above capacity, one needs an order that works
uniformly over all block lengths. Continuity at each fixed block length does
not provide such a choice: the required order may approach one as the block
grows, leaving no positive exponential rate. The missing ingredient is
therefore a uniform comparison with the corresponding ordinary information
quantities.

We establish this comparison for both the sandwiched R\'enyi coherent
information and Holevo quantity
(Proposition~\ref{prop:finite-block-continuity}). After normalization by the
block length, their excess over the respective capacities is bounded by an
explicit correction that vanishes as the R\'enyi order decreases to one.
The correction depends only on the order and the local channel dimensions,
allowing arbitrary correlations across channel uses. Consequently, the
regularized R\'enyi quantities converge to the quantum and classical
capacities, respectively (Theorem~\ref{thm:main}). Combined with the Arimoto
bounds, this proves exponential strong converses for unassisted entanglement
generation and classical communication over arbitrary finite-dimensional
memoryless quantum channels (Theorem~\ref{thm:operational}).

Our proof begins with one-shot integral representations expressing the
R\'enyi quantities through ordinary entropy differences and mutual information
evaluated on tilted states. These states need not lie in the channel's
Stinespring image, preventing a direct comparison with optimization over
admissible inputs. We control their local support violations uniformly in the
number of channel uses and restore support one use at a time. Adapting the
Leung--Smith telescoping argument~\cite{leung2009continuity}, we bound each
entropy change using conditional-entropy
continuity~\cite{alicki2004continuity,winter2016tight} with only local
dimension factors. The resulting total correction is linear in the block
length, and its value per channel use vanishes as the R\'enyi order approaches
one.

Section~\ref{sec:results} presents the main results and a detailed proof outline.
Section~\ref{sec:integral-representations} derives the integral representations,
Section~\ref{sec:projection} establishes the local support estimates and
restoration procedure, and Section~\ref{sec:continuity} proves the uniform
bounds and continuity of the R\'enyi capacities.
Section~\ref{sec:conclusion} concludes. The appendices provide the proofs of
the operational strong-converse exponent characterizations and discuss other
quantum communication criteria.

\section{Proof Outline and Main Results}
\label{sec:results}

\subsection{The Arimoto route to a strong converse}\label{sec:arimoto}

The Arimoto method converts a R\'enyi information bound into an exponential
converse. It uses the sandwiched R\'enyi divergence
$\widetilde D_\alpha$, formally introduced in
Section~\ref{sec:integral-representations}. For quantum communication we focus
on entanglement generation. Fix a
finite-dimensional channel $\mathcal N:A\to B$. An $(n,M)$ code prepares a
state $\rho_{RA^n}$ locally, sends $A^n$ through
$\mathcal N^{\otimes n}$, and applies a decoder
$\mathcal D:B^n\to B_0$, where $\dim R=\dim B_0=M$. Let
$\ket{\Phi}_{RB_0}=M^{-1/2}\sum_{j=1}^M\ket j_R\ket j_{B_0}$ and
$\Phi_{RB_0}=\ket{\Phi}\!\bra{\Phi}_{RB_0}$. Its squared generation fidelity is
\begin{equation}
 F_{\mathrm{eg}}=\Tr\!\left[
 \Phi_{RB_0}(\id_R\otimes\mathcal D\circ\mathcal N^{\otimes n})(\rho_{RA^n})
 \right].
 \label{eq:generation-fidelity}
\end{equation}
The retained marginal $\rho_R$ is unrestricted. In terms of the sandwiched
R\'enyi divergence, define the sandwiched R\'enyi coherent information as
\begin{equation}
 \RenyiCoherentInformation{\alpha}(\mathcal N^{\otimes n})
 :=\sup_{\rho_{A^n}}\inf_{\sigma_{B^n}}
       \widetilde D_\alpha\!\left(
       (\id_R\otimes\mathcal N^{\otimes n})(\psi^\rho)
       \middle\| I_R\otimes\sigma_{B^n}\right),
 \label{eq:coherent-information}
\end{equation}
where $\psi^\rho_{RA^n}$ is any purification of $\rho_{A^n}$ and the
infimum is over density operators on $B^n$. Reference isometries leave the
divergence unchanged, so the value does not depend on the choice of
purification. A binary-test data-processing argument
\cite{sharmawarsi2013fundamental,Wilde3,wildewinter2014erasure} gives the
one-shot converse
\begin{equation}
 F_{\mathrm{eg}}\le
 2^{-\frac{\alpha-1}{\alpha}
       [\log M-\RenyiCoherentInformation{\alpha}
          (\mathcal N^{\otimes n})]},\qquad\forall\,\alpha>1.
 \label{eq:code-bound}
\end{equation}
Indeed, after decoding, the test for $\Phi_{RB_0}$ accepts the code state with
probability $F_{\mathrm{eg}}$ and $I_R\otimes\sigma_{B^n}$ with weight $1/M$;
data processing for $\widetilde D_\alpha$ and optimization over the input and
$\sigma_{B^n}$ give Eq.~\eqref{eq:code-bound}.

Regularizing over tensor powers defines the R\'enyi quantum capacity of order
$\alpha$,
\begin{equation}
 \RenyiQuantumCapacity{\alpha}(\mathcal N)
 :=\lim_{n\to\infty}\frac1n
       \RenyiCoherentInformation{\alpha}(\mathcal N^{\otimes n}).
 \label{eq:regularization}
\end{equation}
The optimized conditional R\'enyi entropy is additive on product states.
Equivalently for the channel optimization, product inputs give
\begin{equation} \label{eq:regularized-supremum}
 \RenyiCoherentInformation{\alpha}(\mathcal N^{\otimes(n+m)})
 \ge \RenyiCoherentInformation{\alpha}(\mathcal N^{\otimes n})
      +\RenyiCoherentInformation{\alpha}(\mathcal N^{\otimes m}).
\end{equation}
Fekete's lemma therefore shows that the limit in Eq.~\eqref{eq:regularization}
exists and $\RenyiCoherentInformation{\alpha}(\mathcal N^{\otimes n})
\le n\RenyiQuantumCapacity{\alpha}(\mathcal N)$. Writing $M=2^{nr}$ defines
the rate $r=n^{-1}\log M$. Equation~\eqref{eq:code-bound} and the preceding
inequality then give the Arimoto converse.

\begin{arimotoproposition}[Arimoto converse for entanglement generation]
\label{prop:arimoto-quantum}
For every $\alpha>1$, the fidelity of every $(n,2^{nr})$
entanglement-generation code satisfies
\begin{equation}
 F_{\mathrm{eg}}\le
 2^{-n\frac{\alpha-1}{\alpha}
       [r-\RenyiQuantumCapacity{\alpha}(\mathcal N)]}.
 \label{eq:regularized-arimoto-bound}
\end{equation}
\end{arimotoproposition}
\noindent Proposition~\ref{prop:arimoto-quantum} is folklore, but a formal
proof is given in Appendix~\ref{app:one-shot} for completeness.

\smallskip
Let us for a moment assume that we had the asymptotic continuity statement
\begin{equation}
 \boxed{\displaystyle
 \lim_{\alpha\searrow1}\RenyiQuantumCapacity{\alpha}(\mathcal N)
 \overset{?}{=}\QuantumCapacity(\mathcal N).}
 \label{eq:renyi-quantum-continuity-goal}
\end{equation}
Then, for every rate $r>\QuantumCapacity(\mathcal N)$, we could choose a fixed
$\alpha>1$ for which the exponent in Eq.~\eqref{eq:regularized-arimoto-bound}
is positive, and the fidelity would decay exponentially with the number of
channel uses. Hence, the only missing ingredient for an exponential strong
converse is indeed Eq.~\eqref{eq:renyi-quantum-continuity-goal}, and the goal
of the remainder is to establish exactly this continuity statement.

The quantum statement above uses squared fidelity. Entanglement transmission
and reference-assisted worst-case subspace fidelity are bounded by the
entanglement-generation fidelity, while ordinary worst-case pure-state
fidelity is at most twice as large. Hence the same exponential strong converse
property follows for the standard unassisted definitions.

\medskip
For classical communication, an $(n,M)$ code uses a POVM
$\{D_m\}_{m=1}^M$ on $B^n$ and has average success probability
\begin{equation}
 P_{\mathrm s}=\frac1M\sum_{m=1}^M
       \Tr\!\left[\mathcal N^{\otimes n}(\rho_{A^n}^m)D_m\right].
 \label{holevo:eq:success}
\end{equation}
For an input ensemble $\{p_x,\rho_A^x\}_x$, we define, consistently with
Eq.~\eqref{eq:intro-classical-capacity}, the c-q output state
$\omega_{XB} :=\sum_xp_x\ket{x}\!\bra{x}\otimes
\mathcal M(\rho_A^x)$. For any c-q state $\omega_{XB}$, define its sandwiched
R\'enyi mutual information of order $\alpha>1$ by
\begin{equation}
 \RenyiMutualInformation{\alpha}(X; B)_\omega
 :=\inf_{\sigma_B}
       \widetilde D_\alpha\!\left(
       \omega_{XB}
       \,\middle\|\,
       \omega_X\otimes\sigma_B\right),
 \label{holevo:eq:renyi-mutual-information}
\end{equation}
where the infimum is over density operators on $B$. The R\'enyi Holevo
quantity and the R\'enyi classical capacity of order $\alpha$ are then given by
\begin{equation}
 \RenyiHolevoQuantity{\alpha}(\mathcal M)
 :=\sup_{\{p_x,\rho_A^x\}_x}
       \RenyiMutualInformation{\alpha}(X; B)_{\omega} \ , \qquad  \RenyiClassicalCapacity{\alpha}(\mathcal N)
 :=\lim_{n\to\infty}\frac1n
          \RenyiHolevoQuantity{\alpha}(\mathcal N^{\otimes n}).
 \label{holevo:eq:definition} 
\end{equation}
A similar argument as the one used above yields the
following converse statement.

\begin{arimotoproposition}[Arimoto converse for classical communication]
\label{prop:arimoto-classical}
For every $\alpha>1$, the success probability of every $(n,2^{nr})$ classical
code satisfies
\begin{equation}
 P_{\mathrm s}\le
 2^{-n\frac{\alpha-1}{\alpha}
       [r-\RenyiClassicalCapacity{\alpha}(\mathcal N)]}.
 \label{holevo:eq:regularized-arimoto-bound}
\end{equation}
\end{arimotoproposition}
\noindent A proof of Proposition~\ref{prop:arimoto-classical} is given
in Appendix~\ref{app:classical-operational} for completeness. The missing
asymptotic statement, 
\begin{equation}
 \boxed{\displaystyle
 \lim_{\alpha\searrow1}\RenyiClassicalCapacity{\alpha}(\mathcal N)
 \overset{?}{=}\ClassicalCapacity(\mathcal N) ,}
 \label{holevo:eq:renyi-classical-continuity-goal}
\end{equation}
would make the exponent in
Eq.~\eqref{holevo:eq:regularized-arimoto-bound} positive at every rate above
$\ClassicalCapacity(\mathcal N)$.

\subsection{One-shot integral representations}

It remains only to prove the two continuity statements in
Eqs.~\eqref{eq:renyi-quantum-continuity-goal}
and~\eqref{holevo:eq:renyi-classical-continuity-goal}. Recall that
$\CoherentInformation(\mathcal M)$ is the maximum of $H(B)-H(E)$ over
Stinespring outputs, whereas $\HolevoQuantity(\mathcal M)$ is the maximum of
$I(X \!: \! B)$ over c-q channel output states. Our main one-shot ingredient
represents their R\'enyi counterparts as averages of these same quantities.
For each input state or input family, the quantum tilted states
$\zeta_{BE}(u)$ and the c-q tilted states $\xi_{XBE}(u)$ are defined in
Eqs.~\eqref{eq:tilted-state} and~\eqref{holevo:eq:classical-extensions},
respectively, in Section~\ref{sec:integral-representations}. The integral representations are in terms of these states.

\begin{theorem}[Integral representations]
\label{thm:integral-representations}
For every finite-dimensional channel $\mathcal M$ and every $\alpha>1$, set
$s_\alpha=(\alpha-1)/\alpha$. Then
\begin{equation}
 \RenyiCoherentInformation{\alpha}(\mathcal M)
 =\sup_{\rho_A}\left\{\frac1{s_\alpha}\int_0^{s_\alpha}
       \bigl[H(B)_{\zeta(u)}-H(E)_{\zeta(u)}\bigr]\,\mathrm{d}u\right\},
 \label{eq:integral-coherent}
\end{equation}
and
\begin{equation}
 \RenyiHolevoQuantity{\alpha}(\mathcal M)
 =\sup_{\{\rho_A^x\}_x}\left\{\frac1{s_\alpha}\int_0^{s_\alpha}
       I(X:B)_{\xi(u)}\,\mathrm{d}u\right\}.
 \label{holevo:eq:integral}
\end{equation}
The second supremum is over finite families of pure input states. For each
family, the probability distribution defining $\xi(u)$ is optimized
separately at each value of $u$.
\end{theorem}

\noindent The theorem is the main new one-shot ingredient. The proof
of Theorem~\ref{thm:integral-representations} is given in
Section~\ref{sec:integral-representations}.

Fix a Stinespring isometry $U:A\to B\otimes E$ and write
$\mathcal N^{\mathrm c}:A\to E$ for the complementary channel.
Suppose for the moment that every tilted state $\zeta_{BE}(u)$ in
\eqref{eq:integral-coherent} were supported on
$\operatorname{im}U$. It would then be the Stinespring output of
an admissible input, so
\begin{equation}
 H(B)_\zeta-H(E)_\zeta
 \le \CoherentInformation(\mathcal N).
 \label{eq:ideal-support-bound}
\end{equation}
Substitution into the integral representation would prove
$\RenyiCoherentInformation{\alpha}(\mathcal N)\le
\CoherentInformation(\mathcal N)$. Applying the same one-shot argument to
$\mathcal N^{\otimes n}$ for every $n$ and regularizing would also give
$\RenyiQuantumCapacity{\alpha}(\mathcal N)\le\QuantumCapacity(\mathcal N)$,
and the argument would be finished. However, this
idealised statement cannot hold in general. R\'enyi coherent information
is monotone in the order, so $\RenyiCoherentInformation{\alpha}$ is normally
larger than $\CoherentInformation$ for $\alpha>1$. Indeed, the tilt acts separately on
the output and environment marginals
and need not preserve the correlated Stinespring subspace. The c-q
tilted ensemble has the same obstruction. The remaining task is therefore
quantitative: show that for $\alpha$ close to one each tilted state has only a
small component outside the Stinespring image.

\subsection{One channel image at a time}

Consider the $n$-fold channel
$\mathcal N^{\otimes n}:A^n\to B^n$, with input systems
$A_1,\ldots,A_n$, output systems $B_1,\ldots,B_n$, and environment systems
$E_1,\ldots,E_n$. Its Stinespring isometry is
$U^{\otimes n}:A^n\to B^nE^n$. We write $B_{-i}$ and $E_{-i}$ for the
collections of all output and environment systems other than $B_i$ and $E_i$.
For a quantum input $\rho_{A^n}$, let
$\zeta_{B^nE^n}(u)$ be the tilted state obtained by applying
Eq.~\eqref{eq:integral-coherent} to $\mathcal N^{\otimes n}$. For a finite
family of pure inputs $\{\rho_{A^n}^x\}_x$, let
$\xi_{XB^nE^n}(u)$ be the corresponding c-q tilted state from
Eq.~\eqref{holevo:eq:integral}, where $X$ labels the family.

For each $i$, let $P_i$ act as $UU^\dagger$ on $B_iE_i$ and as the identity
on $B_{-i}E_{-i}$. We show that the probability of violating this support
constraint at any one channel use is of order $u^2$. Specifically,
Lemma~\ref{lem:local-support} yields the two local support bounds
\begin{align}
 \Tr[(I-P_i)\zeta(u)]
 &\le2u^2\left(\sqrt{\dim B}
       +(\dim B)^{u/2}\sqrt{\dim E}\right)^2,
 \label{eq:local-support-outline-quantum}\\
 \Tr[(I_X\otimes(I-P_i))\xi(u)]
 &\le2u^2\left(\sqrt{\dim B}
       +(\dim B)^{u/2}\sqrt{\dim E}\right)^2.
 \label{eq:local-support-outline}
\end{align}
Crucially, the dimensions of the other $n-1$ channel uses do not appear.
This estimate remains valid for correlated inputs and, in the classical case,
independently of the ensemble size.

We then restore the constraints $P_1,\ldots,P_n$ sequentially. Starting from
$\zeta^{(0)}=\zeta(u)$ in the quantum case or
$\xi^{(0)}=\xi(u)$ in the classical case, at step $i$ we project the pair
$B_iE_i$ onto $\operatorname{im}U$ and replace the rejected component by a
fixed state in that subspace. If the right sides of
Eqs.~\eqref{eq:local-support-outline-quantum} and
\eqref{eq:local-support-outline} are at most $\delta^2/2$, consecutive
states are within trace distance $\delta$. The operation acts only on
$B_iE_i$, so it preserves the marginal on $B_{-i}E_{-i}$ in the quantum
case and on $XB_{-i}E_{-i}$ in the classical case. The entropy chain rule
therefore gives, for example,
\begin{equation}
 H(B^n)_{\zeta^{(i)}}-H(B^n)_{\zeta^{(i-1)}}
 =H(B_i|B_{-i})_{\zeta^{(i)}}
  -H(B_i|B_{-i})_{\zeta^{(i-1)}}.
\label{eq:outline-conditional-entropy-telescoping}
\end{equation}
A continuity bound for conditional entropy
\cite{berta2026sharpcontinuity,chengliu2026sharp} controls this
difference using only $\delta$ and $\dim B_i$, irrespective of the size of
$B_{-i}$. The same argument applies to $E_i$ for coherent information. For
classical communication, the mutual-information chain rule produces two
conditional-entropy differences with first system $B_i$, one conditioned on
$B_{-i}$ and one on $XB_{-i}$. Summing these local estimates over
$i=1,\ldots,n$ is the Leung--Smith telescoping argument
\cite{leung2009continuity} and gives a correction linear in $n$. The final
state lies in $\operatorname{im}(U^{\otimes n})$, where the ordinary
coherent- or Holevo-information optimization applies.
Section~\ref{sec:projection} contains the full argument.

\subsection{Continuity of the R\'enyi capacities and the exponential strong
converse}

Let us summarize how the ingredients fit together. Applying
Theorem~\ref{thm:integral-representations} to
$\mathcal N^{\otimes n}$ expresses the order-$\alpha$ coherent and Holevo
information as integrals of ordinary entropy differences and mutual
information evaluated on the tilted states.
Equations~\eqref{eq:local-support-outline-quantum} and
\eqref{eq:local-support-outline} show that these states leave each local Stinespring image by only
$O((\alpha-1)^2)$, uniformly in the input, the number of channel uses and,
in the classical case, the ensemble size. Restoring the $n$ local images one
at a time and applying the conditional-entropy continuity bound at each step
therefore changes the relevant information quantity by at most
$n\varepsilon(\alpha)$, where $\varepsilon(\alpha)\to0$ as
$\alpha\searrow1$. The repaired state lies in the image of
$U^{\otimes n}$, so its information is bounded by the corresponding
order-one quantity of $\mathcal N^{\otimes n}$. Dividing by $n$ and
regularizing proves the desired continuity.

\begin{theorem}[Continuity of R\'enyi capacities]\label{thm:main}
For every finite-dimensional quantum channel $\mathcal N$,
\begin{equation}
 \lim_{\alpha\searrow1}\RenyiQuantumCapacity{\alpha}(\mathcal N)
 =\QuantumCapacity(\mathcal N)
 \quad\text{and}\quad
 \lim_{\alpha\searrow1}\RenyiClassicalCapacity{\alpha}(\mathcal N)
 =\ClassicalCapacity(\mathcal N).
 \label{eq:capacity-continuity}
\end{equation}
\end{theorem}

\noindent The stronger finite-block inequalities, including an explicit continuity
modulus in the R\'enyi order, are stated in
Proposition~\ref{prop:finite-block-continuity} below.
A formal proof of Theorem~\ref{thm:main} is given in
Section~\ref{sec:continuity}.

Theorem~\ref{thm:main} supplies the missing input to the Arimoto bounds in
Propositions~\ref{prop:arimoto-quantum} and~\ref{prop:arimoto-classical}.
Choosing $\alpha>1$ sufficiently close to one immediately gives exponential
decay above the respective capacities.
Appendices~\ref{app:quantum-exponent} and~\ref{app:classical-exponent} give the
remaining operational arguments. The matching achievability bound is proved in
Appendix~\ref{app:eg-upper} for entanglement generation and obtained from
Mosonyi and Ogawa~\cite[Theorem~VI.1]{mosonyi2017strong} in
Appendix~\ref{app:classical-upper} for classical communication, showing that
the resulting exponents are exact.

For codes with $M\ge2^{nr}$, let
$F_{\mathrm{eg}}^{\star}(n,r)$ and $P_{\mathrm s}^{\star}(n,r)$ denote the
supremal entanglement-generation fidelity and classical success probability,
respectively. The corresponding strong-converse exponents are
\begin{align}
 E_{\mathrm{sc,eg}}(r,\mathcal N)
 &:=\lim_{n\to\infty}-\frac1n\log F_{\mathrm{eg}}^{\star}(n,r),
 \label{eq:exponent-definition}\\
 E_{\mathrm{sc,c}}(r,\mathcal N)
 &:=\lim_{n\to\infty}-\frac1n\log P_{\mathrm s}^{\star}(n,r).
 \label{holevo:eq:exponent-definition}
\end{align}
Both limits exist. Indeed, tensoring independent codes makes either optimal
performance supermultiplicative. On the other hand, for
$M_n=\lceil2^{nr}\rceil$, product preparation has entanglement-generation
fidelity $1/M_n$, and always guessing one message has classical success
probability $1/M_n$. Thus the negative logarithm of either optimal
performance is a finite, nonnegative, subadditive sequence. Fekete's lemma
gives the limits in the two definitions above and places them in $[0,r]$.

\begin{theorem}[Exponential strong converses]\label{thm:operational}
For every finite-dimensional quantum channel $\mathcal N$ and every $r\ge0$,
\begin{align}
 E_{\mathrm{sc,eg}}(r,\mathcal N)
 &=\sup_{\alpha>1}\frac{\alpha-1}{\alpha}
       \bigl[r-\RenyiQuantumCapacity{\alpha}(\mathcal N)\bigr],
 \label{eq:exponent-bound}\\
E_{\mathrm{sc,c}}(r,\mathcal N)
&=\sup_{\alpha>1}\frac{\alpha-1}{\alpha}
       \bigl[r-\RenyiClassicalCapacity{\alpha}(\mathcal N)\bigr].
\label{holevo:eq:exponent-bound}
\end{align}
The first exponent is strictly positive for $r>\QuantumCapacity(\mathcal N)$, and the
second is strictly positive for $r>\ClassicalCapacity(\mathcal N)$.
\end{theorem}

\noindent A formal proof of Theorem~\ref{thm:operational} is assembled at the end of
Appendix~\ref{app:classical-exponent} from the quantum and classical arguments
developed in Appendices~\ref{app:quantum-exponent}
and~\ref{app:classical-exponent}, respectively.

Our novelty claim does not rest on the form of either exponent formula in
Theorem~\ref{thm:operational}. We include the entanglement-generation
achievability proof because we could not locate one in the literature. The
central new result is that both exponents are strictly positive at every rate
exceeding the corresponding capacity. Evaluating the regularized formulas for
concrete nonadditive channels remains an important problem.

\section{Integral representations}
\label{sec:integral-representations}

This section proves the one-shot identities in
Theorem~\ref{thm:integral-representations}.

\subsection{R\'enyi divergence preliminaries}

The logarithm is taken to the binary basis throughout. For a state $\rho$ and a positive semidefinite comparison operator $\sigma$,
the sandwiched R\'enyi divergence at order $\alpha>1$ is
\cite{MDS+13,Wilde3}
\begin{equation}
 \widetilde D_\alpha(\rho\|\sigma)
 =\frac{1}{\alpha-1}\log\Tr\!\left[
 \left(\sigma^{\frac{1-\alpha}{2\alpha}}\rho
             \sigma^{\frac{1-\alpha}{2\alpha}}\right)^\alpha\right].
 \label{eq:divergence}
\end{equation}
It is infinite unless $\supp\rho\subseteq\supp\sigma$; negative powers are
evaluated on the support. The order-one quantity is the Umegaki relative
entropy
\begin{equation}
 D(\rho\|\sigma)=\Tr[\rho(\log\rho-\log\sigma)],
\label{eq:umegaki-relative-entropy}
\end{equation}
with the same support convention and value $+\infty$ when
$\supp\rho\nsubseteq\supp\sigma$.
We collect the properties of the divergence used throughout the paper.

\begin{samepage}
\begin{fact}[Basic properties of the sandwiched R\'enyi divergence]
\label{fact:sandwiched-data-processing}
For every state $\rho$, positive semidefinite operator $\sigma$, quantum
channel $\mathcal E$, and $1<\alpha\le\beta$, one has
\begin{equation}
 \widetilde D_\alpha\!\left(\mathcal E(\rho)\middle\|\mathcal E(\sigma)\right)
 \le \widetilde D_\alpha(\rho\|\sigma),
 \qquad
 D(\rho\|\sigma)\le\widetilde D_\alpha(\rho\|\sigma)
 \le\widetilde D_\beta(\rho\|\sigma),
 \label{eq:fact-sandwiched-basic-properties}
\end{equation}
and
\begin{equation}
 \lim_{\alpha\searrow1}\widetilde D_\alpha(\rho\|\sigma)
 =D(\rho\|\sigma).
 \label{eq:fact-sandwiched-order-one-limit}
\end{equation}
For data processing in the range $1<\alpha\le2$ used below, see
Refs.~\cite[Theorem~6]{MDS+13} and~\cite[Corollary~6]{Wilde3}; the extension
to all $\alpha>1$ is Ref.~\cite[Theorem~6]{beigi2013data}. Monotonicity in the
order and the order-one limit are
Refs.~\cite[Theorem~7]{MDS+13} and~\cite[Theorem~5]{MDS+13}, respectively.
\end{fact}
\end{samepage}

The von Neumann entropy of a state $\rho$ is given by $H(\rho) := -\log \Tr[\rho \log \rho]$ and if we have a bipartite state $\rho_{AB}$ we use the shorthand $H(A)_{\rho}$ to denote the entropy of the marginal $\rho_A$. The mutual information is defind as $I(A\!:\!B) := D(\rho_{AB} \| \rho_A \otimes \rho_B)$. The optimized sandwiched R\'enyi
conditional entropy is defined by
\begin{equation}
 \widetilde H_\alpha^\uparrow(A|B)_\rho
 :=-\inf_{\sigma_B}\widetilde D_\alpha\!\left(
       \rho_{AB}\,\middle\|\,I_A\otimes\sigma_B\right).
 \label{eq:conditional-renyi-shorthand}
\end{equation}

We use the following variational identity from M\"uller-Lennert et
al.~\cite{MDS+13}, stated directly in the notation needed here.

\begin{fact}[Minimax representation]
\label{fact:mds-variational}
Let $\omega_{RBE}$ be pure, let $\alpha>1$, and set
$s_\alpha=(\alpha-1)/\alpha$. Then
\begin{equation}
 2^{-s_\alpha\widetilde H_\alpha^\uparrow(R|B)_\omega}
 =\inf_{\sigma_B>0}\sup_{\tau_E}
     \Tr[\omega_{BE}(\sigma_B^{-s_\alpha}\otimes\tau_E^{s_\alpha})],
 \label{eq:fact-mds-variational}
\end{equation}
where $\sigma_B$ ranges over strictly positive density operators and
$\tau_E$ over density operators.
This is Ref.~\cite[Eq.~(19)]{MDS+13}.
\end{fact}

We also use the optimizer characterization of Hayashi and
Tomamichel~\cite{hayashitomamichel15c}.

\begin{fact}[Optimizer and fixed point]
\label{fact:ht-fixed-point}
Let $\rho_{AB}$ be a state, let $\tau_A$ be a density operator with
$\supp\rho_A\subseteq\supp\tau_A$, and let $\alpha>1$. The optimization
\begin{equation}
 \inf_{\sigma_B}\widetilde D_\alpha
       (\rho_{AB}\|\tau_A\otimes\sigma_B)
 \label{eq:fact-ht-optimization}
\end{equation}
has a unique minimizer $\sigma_B^\star$ with
$\supp\sigma_B^\star=\supp\rho_B$. This minimizer satisfies the fixed-point
relations
\begin{equation}
 Z_{AB}=\left[(\tau_A\otimes\sigma_B^\star)^{\frac{1-\alpha}{2\alpha}}
                    \rho_{AB}
              (\tau_A\otimes\sigma_B^\star)^{\frac{1-\alpha}{2\alpha}}
             \right]^\alpha,
 \qquad
 \sigma_B^\star=\frac{\Tr_A[Z_{AB}]}{\Tr[Z_{AB}]}.
 \label{eq:fact-ht-fixed-point}
\end{equation}
This is Ref.~\cite[Lemma~5]{hayashitomamichel15c}.
\end{fact}

\subsection{Quantum coherent information}\label{sec:integral-quantum}

The construction below is first stated for faithful $B$ and $E$ marginals.
This is no restriction: the general case follows by restricting $B$ and $E$
to their respective supports, as explained at the end of this section.  The
tilted state is useful because its entropy difference represents the derivative
of the R\'enyi coherent information, turning the R\'enyi quantity into the von
Neumann entropies to which the continuity bounds of
Section~\ref{sec:projection} apply.

\begin{lemma}[Tilted state]\label{lem:tilted-state}
Let $\omega_{RBE}$ be a pure state with $\omega_B>0$ and $\omega_E>0$.
For each $0<s<1$,
\begin{equation}
 2^{-s\widetilde H_{1/(1-s)}^\uparrow(R|B)_\omega}
 =\min_{\sigma_B>0}\max_{\tau_E}
     \Tr[\omega_{BE}(\sigma_B^{-s}\otimes\tau_E^s)],
 \label{eq:variational}
\end{equation}
where both optimizations are over density operators.
There is a unique strictly positive minimizing $\sigma_B^\star$; for this
optimizer, the maximizing $\tau_E^\star$ is also unique and strictly positive.
We define the normalized tilted state associated with them by
\begin{equation}
 \zeta_{BE}=
 \frac{((\sigma_B^\star)^{-s/2}\otimes(\tau_E^\star)^{s/2})\omega_{BE}
       ((\sigma_B^\star)^{-s/2}\otimes(\tau_E^\star)^{s/2})}
      {\Tr[\omega_{BE}((\sigma_B^\star)^{-s}\otimes(\tau_E^\star)^s)]}
 \label{eq:tilted-state}
\end{equation}
It satisfies
\begin{equation}
 \zeta_B=\sigma_B^\star,\qquad \zeta_E=\tau_E^\star.
 \label{eq:tilted-marginals}
\end{equation}
\end{lemma}

\begin{proof}
Fact~\ref{fact:mds-variational} gives the variational identity in
Eq.~\eqref{eq:variational}, initially with an infimum and a supremum.
Fact~\ref{fact:ht-fixed-point}, applied with
$\kappa_R=I_R/\dim R$, gives a unique minimizing $\sigma_B^\star$ with support
equal to that of $\omega_B$; the normalization of $I_R$ adds only a constant
to the divergence. Hence $\sigma_B^\star>0$.

For this minimizing state, consider the unnormalized pure vector
$(I_R\otimes(\sigma_B^\star)^{-s/2}\otimes I_E)\ket\omega$.
Its $RB$ and $E$ reductions have the same nonzero eigenvalues.
Its $E$ reduction is
\begin{equation}
 Y_E=\Tr_B[((\sigma_B^\star)^{-s/2}\otimes I_E)\omega_{BE}
                       ((\sigma_B^\star)^{-s/2}\otimes I_E)].
 \label{eq:environment-reduction}
\end{equation}
Schatten duality gives the inner maximization in \eqref{eq:variational}:
\begin{equation}
 \norm{Y_E}_{1/(1-s)}=\max_{\tau_E}\Tr[Y_E\tau_E^s],
 \qquad
 \tau_E^\star=\frac{Y_E^{1/(1-s)}}{\Tr[Y_E^{1/(1-s)}]}.
 \label{eq:environment-optimizer}
\end{equation}
Since $\omega_E>0$, the operator
$Y_E$ is strictly positive, so this maximizer is unique and strictly positive.
Thus both extrema in Eq.~\eqref{eq:variational} are attained.
To identify the marginals, purify $\zeta_{BE}$ by the normalized vector
proportional to
$(I_R\otimes(\sigma_B^\star)^{-s/2}\otimes(\tau_E^\star)^{s/2})\ket\omega$.
Using \eqref{eq:environment-optimizer}, a Schmidt decomposition of
$(I_R\otimes(\sigma_B^\star)^{-s/2}\otimes I_E)\ket\omega$ gives
\begin{equation}
 \zeta_{RB}=
 \frac{\bigl[(I_R\otimes(\sigma_B^\star)^{-s/2})\omega_{RB}
                  (I_R\otimes(\sigma_B^\star)^{-s/2})\bigr]^{1/(1-s)}}
      {\Tr\!\left[\bigl[(I_R\otimes(\sigma_B^\star)^{-s/2})\omega_{RB}
                  (I_R\otimes(\sigma_B^\star)^{-s/2})\bigr]^{1/(1-s)}\right]}.
\end{equation}
Equation~\eqref{eq:fact-ht-fixed-point} states that the minimizing comparison
state $\sigma_B^\star$ equals the $B$-marginal of the normalized sandwiched
power on the right-hand side above. Hence
$\Tr_R[\zeta_{RB}]=\sigma_B^\star$.
On the other hand,
\begin{equation}
 \zeta_E=
 \frac{(\tau_E^\star)^{s/2}Y_E(\tau_E^\star)^{s/2}}
      {\Tr[Y_E(\tau_E^\star)^s]}
 =\tau_E^\star,
\end{equation}
again by \eqref{eq:environment-optimizer}.
\end{proof}

For later use, the trial choices $\tau_E=I_E/\dim E$ and
$\sigma_B=I_B/\dim B$ in \eqref{eq:variational}, together with
$\sigma_B^{-s}\ge I_B$ and $\tau_E^s\le I_E$, give
\begin{equation}
 -s\log\dim E
 \le -s\widetilde H_{1/(1-s)}^\uparrow(R|B)_\omega
 \le s\log\dim B.
 \label{eq:dimension-bounds}
\end{equation}

The exact marginals in \eqref{eq:tilted-marginals} are the reason for
introducing this state: they convert the logarithms of the optimizing
states into ordinary entropies. The required regularity and envelope
identity are another external input from Hayashi and
Tomamichel~\cite{hayashitomamichel15c}.

\begin{fact}[Differentiability of the optimized divergence]
\label{fact:ht-envelope}
Let $\rho_{AB}$ and $\tau_A$ satisfy the support assumption in
Fact~\ref{fact:ht-fixed-point}, and define
\begin{equation}
 G(\alpha)=\inf_{\sigma_B}\widetilde D_\alpha
       (\rho_{AB}\|\tau_A\otimes\sigma_B).
\end{equation}
The function $G$ is continuously differentiable for $\alpha\ge1$, and
its derivative is obtained by differentiating only the R\'enyi order at the
unique minimizing state:
\begin{equation}
 G'(\alpha)=
 \left.\frac{\partial}{\partial\alpha}
 \widetilde D_\alpha(\rho_{AB}\|\tau_A\otimes\sigma_B)
 \right|_{\sigma_B=\sigma_B^\star}.
 \label{eq:fact-ht-envelope}
\end{equation}
This is Ref.~\cite[Proposition~11]{hayashitomamichel15c}.
\end{fact}

\begin{lemma}[Derivative as an entropy difference]\label{lem:derivative}
The function
\begin{equation}
 s\longmapsto-s\widetilde H_{1/(1-s)}^\uparrow(R|B)_\omega
 \label{eq:entropy-function}
\end{equation}
extends to a continuously differentiable function on $[0,1)$, with value
zero at $s=0$ and right derivative $H(B)_\omega-H(E)_\omega$ there.
For every $0<s<1$,
\begin{equation}
 \frac{\mathrm{d}}{\mathrm{d}s}\left[-s\widetilde H_{1/(1-s)}^\uparrow(R|B)_\omega\right]
 = H(B)_\zeta-H(E)_\zeta,
 \label{eq:derivative-identity}
\end{equation}
where $\zeta_{BE}$ is the tilted state in \eqref{eq:tilted-state}
at parameter $s$.
\end{lemma}

\begin{proof}
Apply Fact~\ref{fact:ht-envelope} with state $\omega_{RB}$
and fixed comparison state $I_R/\dim R$. Its support condition is
automatic, and
\begin{equation}
 -\widetilde H_\alpha^\uparrow(R|B)_\omega
 =\inf_{\sigma_B}\widetilde D_\alpha\!\left(
 \omega_{RB}\,\middle\|\,\frac{I_R}{\dim R}\otimes\sigma_B\right)
 -\log\dim R.
\end{equation}
Thus $\alpha\mapsto\widetilde H_\alpha^\uparrow(R|B)_\omega$ is
continuously differentiable through $\alpha=1$. Multiplication by $-s$
and the substitution $\alpha=1/(1-s)$ prove the asserted regularity,
including the extension at $s=0$. The derivative there is
$-H(R|B)_\omega=H(B)_\omega-H(E)_\omega$ by purity.

Fix $0<s<1$ and the optimizing states $\sigma_B^\star,\tau_E^\star$ from
Lemma~\ref{lem:tilted-state}. Equation~\eqref{eq:fact-ht-envelope}
permits us to differentiate with respect to the R\'enyi order while
holding $\sigma_B^\star$ fixed.
For fixed $\sigma_B^\star$, the positive maximizer in
\eqref{eq:environment-optimizer} varies smoothly with the parameter;
its variation contributes zero by stationarity under the trace-one
constraint. Thus no derivative of an optimizer is needed. Keeping both
operators fixed in the derivative below gives
\begin{align}
 \frac{\mathrm{d}}{\mathrm{d}s}\left[-s\widetilde H_{1/(1-s)}^\uparrow(R|B)_\omega\right]
 &=\left.\frac{\mathrm{d}}{\mathrm{d}u}\log\Tr\!\left[
       \omega_{BE}((\sigma_B^\star)^{-u}\otimes(\tau_E^\star)^u)
       \right]\right|_{u=s}
 \label{eq:derivative-envelope-evaluation}\\
 &=\Tr\!\left[\zeta_{BE}
       (I_B\otimes\log\tau_E^\star
        -\log\sigma_B^\star\otimes I_E)\right]
 \label{eq:derivative-log-evaluation}\\
 &=H(B)_\zeta-H(E)_\zeta.
 \label{eq:derivative-evaluation}
\end{align}
The logarithms commute with the corresponding matrix powers in
\eqref{eq:tilted-state}, and \eqref{eq:tilted-marginals} gives the last equality.
\end{proof}

\begin{remark}
\label{rem:fixed-state-derivative}
For fixed states $\rho,\sigma$ with $\supp\rho\subseteq\supp\sigma$ and
$\alpha>1$, Bao, Moosa, and Shehzad
\cite[Eqs.~(3.6), (3.11), and (3.12)]{bao2019holographic} derived
\begin{equation}
 \begin{aligned}
 \alpha^2\frac{\mathrm d}{\mathrm d\alpha}
 \left[\frac{\alpha-1}{\alpha}\widetilde D_\alpha(\rho\|\sigma)\right]
 &=D(\rho_\alpha\|\sigma),\\
 \rho_\alpha
 &:=\frac{(\sigma^{\frac{1-\alpha}{2\alpha}}\rho
                  \sigma^{\frac{1-\alpha}{2\alpha}})^\alpha}
 {\Tr[(\sigma^{\frac{1-\alpha}{2\alpha}}\rho
                  \sigma^{\frac{1-\alpha}{2\alpha}})^\alpha]}.
 \end{aligned}
 \label{eq:fixed-state-renyi-derivative}
\end{equation}
Here $D$ is the ordinary relative entropy. Kibe and Roy
\cite[Theorem~4.1 and Corollary~4.5]{kiberoy2026escort} extended the
corresponding differential and integral representations to sandwiched and
$\alpha$--$z$ R\'enyi divergences on von Neumann algebras.
Lemmas~\ref{lem:derivative} and~\ref{holevo:lem:derivative} can be viewed
as generalizations of the fixed-state identities to the optimized sandwiched
R\'enyi conditional entropy and Holevo information, respectively.
\end{remark}

\subsection{Classical Holevo information}
\label{sec:integral-classical}

We next consider a fixed finite family of states $\{\omega_B^x\}_x$,
with pure extensions $\omega_{BE}^x$ on a common environment $E$.
Write $\RenyiHolevoQuantity{\alpha}(\{\omega_B^x\}_x)$ for its sandwiched
R\'enyi Holevo quantity, optimizing only the probabilities on $x$ in
\eqref{holevo:eq:definition} while keeping the states fixed. Comparison states are restricted to
$\supp\sum_x\omega_B^x$; in the Holevo tilted-state and derivative lemmas,
$\sigma_B>0$ means positivity on that support. This restriction does
not change the infimum: projecting onto the support and replacing its
complement by a state in it fixes every output, so data processing
in Fact~\ref{fact:sandwiched-data-processing} applies. Interior comparison states approximate every boundary state
with finite objective. All operators are embedded in their original
systems when the local estimates are applied.

\begin{lemma}[Tilted states]\label{holevo:lem:tilted-states}
For $0<s<1$ and $\alpha=1/(1-s)$,
\begin{align}
 2^{s\RenyiHolevoQuantity{\alpha}(\{\omega_B^x\}_x)}
 &=\min_{\sigma_B>0}\max_x
       \norm{\sigma_B^{-s/2}\omega_B^x\sigma_B^{-s/2}}_\alpha
 \label{holevo:eq:minimax-radius}\\
 &=\max_p\inf_{\sigma_B>0}\sum_xp_x
       \norm{\sigma_B^{-s/2}\omega_B^x\sigma_B^{-s/2}}_\alpha.
 \label{holevo:eq:minimax}
\end{align}
A saddle pair $(p^\star,\sigma_B^\star)$ exists. For every $x$ with
$p_x^\star>0$,
\begin{equation}
 \norm{(\sigma_B^\star)^{-s/2}\omega_B^x
       (\sigma_B^\star)^{-s/2}}_\alpha
 =2^{s\RenyiHolevoQuantity{\alpha}(\{\omega_B^x\}_x)}.
 \label{holevo:eq:active-letters}
\end{equation}
For every $\sigma_B>0$ and $x$, Schatten duality gives
\begin{equation}
 \norm{\sigma_B^{-s/2}\omega_B^x\sigma_B^{-s/2}}_\alpha
 =\max_{\tau_E}\Tr[\omega_{BE}^x(\sigma_B^{-s}\otimes\tau_E^s)].
 \label{holevo:eq:variational}
\end{equation}
Let $\tau_E^{x,\star}$ be its maximizing state at
$\sigma_B=\sigma_B^\star$.
Define the normalized tilted states, depending on $s$, by
\begin{align}
 \xi_{BE}^x
 &=\frac{((\sigma_B^\star)^{-s/2}\otimes(\tau_E^{x,\star})^{s/2})
                  \omega_{BE}^x
          ((\sigma_B^\star)^{-s/2}\otimes(\tau_E^{x,\star})^{s/2})}
        {\Tr[\omega_{BE}^x((\sigma_B^\star)^{-s}
                  \otimes(\tau_E^{x,\star})^s)]},
 \label{holevo:eq:tilted-state}\\
 \xi_{XBE}&=\sum_xp_x^\star\ket x\!\bra x\otimes\xi_{BE}^x,
 \qquad \tau_{XE}^\star=\sum_xp_x^\star\ket x\!\bra x
                         \otimes\tau_E^{x,\star}.
 \label{holevo:eq:classical-extensions}
\end{align}
Then each $\xi_{BE}^x$ is pure, and
\begin{equation}
 \xi_B=\sigma_B^\star,\qquad \xi_{XE}=\tau_{XE}^\star.
 \label{holevo:eq:tilted-marginals}
\end{equation}
Zero-probability letters may be omitted throughout.
\end{lemma}

\begin{proof}
The information-radius characterization in
Eq.~\eqref{holevo:eq:minimax-radius} is the sandwiched case of
Ref.~\cite[Proposition~IV.2]{mosonyi2017strong}.
The Schmidt-decomposition and Schatten-duality argument in
Lemma~\ref{lem:tilted-state} gives \eqref{holevo:eq:variational}.
Its unique maximizer is
\begin{equation}
 \tau_E^{x,\star}=
 \frac{\left(\Tr_B[((\sigma_B^\star)^{-s/2}\otimes I_E)\omega_{BE}^x
                       ((\sigma_B^\star)^{-s/2}\otimes I_E)]\right)^\alpha}
      {\Tr[((\sigma_B^\star)^{-s/2}\omega_B^x
                       (\sigma_B^\star)^{-s/2})^\alpha]}.
 \label{holevo:eq:environment-optimizer}
\end{equation}
Its support is $\supp\omega_E^x$, independent of $s$ and of the
positive comparison state; strict positivity on this support suffices.
Operator convexity of $t\mapsto t^{-s}$ \cite[Chapter~V]{bhatia1997matrix}
and \eqref{holevo:eq:variational} show that each displayed norm is convex
in $\sigma_B$. Its $\alpha$th power is also convex. Evaluating the
divergence on the classical blocks and applying Sion's theorem
\cite[Corollary~3.5]{sion1958general} gives
\begin{align}
 2^{(\alpha-1)\RenyiHolevoQuantity{\alpha}(\{\omega_B^x\}_x)}
 &=\max_p\inf_{\sigma_B>0}\sum_xp_x
       \norm{\sigma_B^{-s/2}\omega_B^x\sigma_B^{-s/2}}_\alpha^\alpha
 \label{holevo:eq:classical-blocks-ensemble}\\
 &=\inf_{\sigma_B>0}\max_x
       \norm{\sigma_B^{-s/2}\omega_B^x\sigma_B^{-s/2}}_\alpha^\alpha.
 \label{holevo:eq:classical-blocks}
\end{align}
Here the probability simplex is compact and the objective is continuous,
linear in $p$, and convex in $\sigma_B$. Taking an $\alpha$th root and
applying the same minimax theorem to the norms proves
\eqref{holevo:eq:minimax}.

The comparison-state minimum is attained in the interior, because
the trial $\tau_E=I_E/\dim E$ in \eqref{holevo:eq:variational} gives
\begin{equation}
 \max_x\norm{\sigma_B^{-s/2}\omega_B^x\sigma_B^{-s/2}}_\alpha
 \ge\frac{(\dim E)^{-s}}{|X|}
      \Tr\!\left[\left(\sum_x\omega_B^x\right)\sigma_B^{-s}\right],
 \label{holevo:eq:coercivity}
\end{equation}
which diverges at the boundary of the chosen support. The value function
obtained by infimizing over $\sigma_B$ is upper semicontinuous in $p$, so its
maximum is attained.
The resulting saddle property implies \eqref{holevo:eq:active-letters}.

To identify the marginals, observe that this saddle pair also optimizes
\eqref{holevo:eq:classical-blocks}. Indeed, for every comparison state
$\nu_B>0$, Jensen's inequality and the saddle property give
\begin{equation}
 \sum_xp_x^\star\norm{\nu_B^{-s/2}\omega_B^x\nu_B^{-s/2}}_\alpha^\alpha
 \ge\left(\sum_xp_x^\star
       \norm{\nu_B^{-s/2}\omega_B^x\nu_B^{-s/2}}_\alpha\right)^\alpha
 \ge2^{(\alpha-1)\RenyiHolevoQuantity{\alpha}(\{\omega_B^x\}_x)}.
 \label{holevo:eq:same-optimizer}
\end{equation}
Equality holds at $\nu_B=\sigma_B^\star$ by
\eqref{holevo:eq:active-letters}.
Fact~\ref{fact:ht-fixed-point}, applied to
$\omega_{XB}=\sum_xp_x^\star\ket x\!\bra x\otimes\omega_B^x$
with fixed first comparison state $\omega_X$, therefore gives
\begin{equation}
 \sigma_B^\star=
 \frac{\sum_xp_x^\star((\sigma_B^\star)^{-s/2}\omega_B^x
                         (\sigma_B^\star)^{-s/2})^\alpha}
      {\sum_xp_x^\star\Tr[((\sigma_B^\star)^{-s/2}\omega_B^x
                         (\sigma_B^\star)^{-s/2})^\alpha]}.
 \label{holevo:eq:fixed-point}
\end{equation}
The same Schmidt-decomposition calculation as before gives
\begin{equation}
 \xi_B^x=
 \frac{((\sigma_B^\star)^{-s/2}\omega_B^x
              (\sigma_B^\star)^{-s/2})^\alpha}
      {\Tr[((\sigma_B^\star)^{-s/2}\omega_B^x
              (\sigma_B^\star)^{-s/2})^\alpha]},
 \qquad \xi_E^x=\tau_E^{x,\star}.
 \label{holevo:eq:conditional-marginals}
\end{equation}
The denominators are equal on all active letters by
\eqref{holevo:eq:active-letters}. Averaging thus proves
\eqref{holevo:eq:tilted-marginals}.
\end{proof}

Equation~\eqref{holevo:eq:active-letters} ensures equal normalization factors
for all labels with $p_x^\star>0$. Together with the fixed-point relation,
this gives $\sum_xp_x^\star\xi_B^x=\sigma_B^\star$. The same equality enables
the cancellations needed to express the derivative as mutual information in
the next lemma.

For the Holevo quantity, the maximizing probabilities may change with
the order. We consequently use an almost-everywhere derivative identity.
This is sufficient for integration; in fact, it also gives a uniform
Lipschitz bound for each fixed finite output family.

\begin{lemma}[Derivative as Holevo information]\label{holevo:lem:derivative}
The function $s\mapsto s\RenyiHolevoQuantity{1/(1-s)}(\{\omega_B^x\}_x)$,
extended by zero at $s=0$, is Lipschitz on $[0,1)$ with constant
$\log\dim B$. At every differentiability point in $(0,1)$,
and for any saddle pair in Lemma~\ref{holevo:lem:tilted-states},
\begin{equation}
 \frac{\mathrm{d}}{\mathrm{d}s}[s\RenyiHolevoQuantity{1/(1-s)}(\{\omega_B^x\}_x)]
 =I(X:B)_\xi.
 \label{holevo:eq:derivative-identity}
\end{equation}
\end{lemma}

\begin{proof}
On each compact interval $[a,b]\subset(0,1)$,
\eqref{holevo:eq:coercivity} diverges uniformly at the boundary, whereas
the maximally mixed comparison state bounds the minimum by $(\dim B)^b$.
All minimizing comparison states therefore lie in a common compact
subset of the interior of their fixed support. The norms in
\eqref{holevo:eq:minimax} are smooth in $s$ there, with uniformly bounded
derivatives: each matrix has constant rank. Their finite maximum and
minimum, and hence the logarithm of the optimized value, are locally
Lipschitz. The same trial comparison state gives
\begin{equation}
 0\le s\RenyiHolevoQuantity{1/(1-s)}(\{\omega_B^x\}_x)\le s\log\dim B,
 \label{holevo:eq:dimension-bounds}
\end{equation}
so the extension at zero is continuous.

Fix a saddle pair $(p^\star,\sigma_B^\star)$ at $s$, and keep
$\omega_{XB}=\sum_xp_x^\star\ket x\!\bra x\otimes\omega_B^x$ fixed as the order
varies. By \eqref{holevo:eq:same-optimizer},
\begin{equation}
 u\RenyiHolevoQuantity{1/(1-u)}(\{\omega_B^x\}_x)
 \ge u\inf_{\nu_B}\widetilde D_{1/(1-u)}
       (\omega_{XB}\|\omega_X\otimes\nu_B),
 \label{holevo:eq:fixed-ensemble}
\end{equation}
with equality at $u=s$. The right side is continuously differentiable,
and its derivative is obtained with $\nu_B=\sigma_B^\star$ fixed, by
Fact~\ref{fact:ht-envelope}. At a differentiability
point of the left side the derivatives agree, since their difference
has a minimum zero at $s$. No differentiability of the optimizing
probabilities is required.

The right side of \eqref{holevo:eq:fixed-ensemble} is the infimum of
\begin{equation}
 (1-u)\log\sum_xp_x^\star
   \norm{\nu_B^{-u/2}\omega_B^x\nu_B^{-u/2}}_{1/(1-u)}^{1/(1-u)}.
 \label{holevo:eq:fixed-ensemble-objective}
\end{equation}
At $u=s$ and $\nu_B=\sigma_B^\star$, all active norms are equal by
\eqref{holevo:eq:active-letters}. The derivatives of the outer factor
and of the exponent $1/(1-u)$ cancel. Thus
\begin{align}
 \frac{\mathrm{d}}{\mathrm{d}s}[s\RenyiHolevoQuantity{1/(1-s)}(\{\omega_B^x\}_x)]
 &=\sum_xp_x^\star\left.\frac{\mathrm{d}}{\mathrm{d}u}\log
   \norm{(\sigma_B^\star)^{-u/2}\omega_B^x
           (\sigma_B^\star)^{-u/2}}_{1/(1-u)}
   \right|_{u=s}
 \label{holevo:eq:derivative-envelope-evaluation}\\
 &=\sum_xp_x^\star\Tr\!\left[\xi_{BE}^x
     (I_B\otimes\log\tau_E^{x,\star}
      -\log\sigma_B^\star\otimes I_E)\right]
 \label{holevo:eq:derivative-log-evaluation}\\
 &=H(B)_\xi-\sum_xp_x^\star H(E)_{\xi^x}
 =I(X:B)_\xi.
 \label{holevo:eq:derivative-evaluation}
\end{align}
The second equality differentiates \eqref{holevo:eq:variational},
holding its unique maximizer fixed, just as in Lemma~\ref{lem:derivative}.
The explicit formula \eqref{holevo:eq:environment-optimizer} is smooth
on its fixed support, so its variation contributes zero by stationarity.
The last line uses \eqref{holevo:eq:tilted-marginals} and purity of
each $\xi_{BE}^x$, which gives $H(E)_{\xi^x}=H(B)_{\xi^x}$.

Finally, $0\le I(X:B)_\xi\le\log\dim B$. Integrating the resulting
almost-everywhere derivative bound on $[a,b]\subset(0,1)$ gives
the asserted Lipschitz constant, and continuity at zero extends it to
$[0,1)$. In particular, the function is absolutely continuous on
every $[0,s]$ with $s<1$.
\end{proof}

\subsection{Proof of Theorem~\ref{thm:integral-representations}}
\label{sec:proof-integral-representations}

\begin{proof}
For a fixed quantum input, integrate
Eq.~\eqref{eq:derivative-identity} from zero to
$s_\alpha=(\alpha-1)/\alpha$. This gives the asserted entropy integral. If either
marginal is singular, restrict $B$ and $E$ to the supports of $\omega_B$ and
$\omega_E$, respectively; the state, entropies and optimized divergence are
unchanged, and Lemmas~\ref{lem:tilted-state} and~\ref{lem:derivative}
apply on these supports. The conditional R\'enyi entropy is invariant under
isometries on the purifying system, and the tilted state depends only on the
Stinespring output $U\rho_AU^\dagger$. Thus neither side depends on the chosen
purification. Taking the supremum over $\rho_A$ proves
Eq.~\eqref{eq:integral-coherent}.

For a fixed finite family of pure inputs, the function in
Lemma~\ref{holevo:lem:derivative} is absolutely continuous. Integrating its
almost-everywhere derivative identity from zero to $s$ gives the asserted
mutual-information integral for that family. At differentiability points the
identity holds for every saddle pair; choices on the null set of remaining
parameters do not affect the integral. Every finite mixed-input ensemble has
a finite pure-state refinement, and forgetting the refinement is a channel on
the classical label. Fact~\ref{fact:sandwiched-data-processing} therefore shows that restricting to pure
input families does not lower the supremum. Taking this supremum proves
Eq.~\eqref{holevo:eq:integral}.
\end{proof}

\section{Projection onto the Stinespring image}
\label{sec:projection}

We now prove the local projection and restoration estimates used to establish
continuity in Section~\ref{sec:continuity}. The matrix-power and restoration
lemmas are shared by the quantum and classical cases; the final entropy
accounting is then applied to coherent information and Holevo information
separately. We
write $h_2(t)=-t\log t-(1-t)\log(1-t)$ for the binary entropy.

\subsection{Local support estimates}\label{sec:projection-local}

Both tilted-state constructions may leave the Stinespring image. The next
matrix estimates control this departure using only a local dimension.
They allow singular states, which is essential because different
ensemble members may have different environment supports.

\begin{lemma}[Conditional matrix powers]\label{lem:conditional-powers}
For any state $\rho_{AB}$ and $0<s\le1$,
\begin{equation}
 \norm{(I_A\otimes\rho_B)^{-s/2}\rho_{AB}^{s/2}}_\infty
 \le(\dim A)^{s/2}.
 \label{eq:power-operator}
\end{equation}
Moreover, for every $-1\le s\le1$,
\begin{equation}
 \left\|\left[(I_A\otimes\rho_B)^{-s/2}\rho_{AB}^{s/2}-I\right]
                \sqrt{\rho_{AB}}\right\|_2
 \le\sqrt{2(\dim A-1)}\,|s|.
 \label{eq:power-difference}
\end{equation}
Negative powers are evaluated on the support and extended by zero.
The bounds are independent of the conditioning dimension $\dim B$
and of the smallest nonzero eigenvalues.
\end{lemma}

\begin{proof}
The reduction inequality
$\rho_{AB}\le(\dim A)(I_A\otimes\rho_B)$
\cite{tomamichel16_book} and L\"owner--Heinz monotonicity give the same
inequality for the $s$th powers, with coefficient $(\dim A)^s$.
Conjugating by
$(I_A\otimes\rho_B)^{-s/2}$ proves \eqref{eq:power-operator}.

For the second bound, let $(r_i,\ket{r_i})$ and
$(m_j,\ket{m_j})$ be the positive eigenpairs
of $\rho_{AB}$ and $I_A\otimes\rho_B$, respectively. Since
$\supp\rho_{AB}\subseteq\supp(I_A\otimes\rho_B)$, the weights
$r_i|\langle r_i|m_j\rangle|^2$ sum to one. Direct expansion, retaining
the indicated order of the matrix factors, gives
\begin{align}
 &\left\|\left[(I_A\otimes\rho_B)^{-s/2}
          \rho_{AB}^{s/2}-I\right]\sqrt{\rho_{AB}}\right\|_2^2
          \nonumber\\
 &\qquad=\sum_{i,j}r_i|\langle r_i|m_j\rangle|^2
       \left[\left(\frac{r_i}{m_j}\right)^{s/2}-1\right]^2.
 \label{eq:spectral-identity}
\end{align}
These are the Nussbaum--Szko\l a spectral weights
\cite{nussbaum2009chernoff}. For $x>0$ and $-1\le s\le1$, convexity of
the exponential gives
\begin{align}
 |x^{s/2}-1|
 &\le e^{|s||\ln x|/2}-1
 \le |s|\bigl(e^{|\ln x|/2}-1\bigr)
 \le |s|\left|\sqrt{x}-\frac1{\sqrt{x}}\right|,
 \label{eq:scalar-bound-first-order}\\
 |x^{s/2}-1|^2&\le s^2(x+x^{-1}-2).
 \label{eq:scalar-bound}
\end{align}
Thus the squared norm is bounded by
\begin{equation}
 s^2\left[\Tr[\rho_{AB}^2(I_A\otimes\rho_B)^{-1}]
                 +\Tr[I_A\otimes\rho_B]-2\right]
 \le2(\dim A-1)s^2.
 \label{eq:power-second-moment}
\end{equation}
For the inverse-ratio term, omitting zero eigenvalues of $\rho_{AB}$
can only reduce the trace $\Tr[I_A\otimes\rho_B]=\dim A$.
For the first term, \eqref{eq:power-operator} at $s=1$ gives
\begin{equation}
 \Tr[\rho_{AB}^2(I_A\otimes\rho_B)^{-1}]
 =\left\|(I_A\otimes\rho_B)^{-1/2}\rho_{AB}\right\|_2^2
 \le\left\|(I_A\otimes\rho_B)^{-1/2}\sqrt{\rho_{AB}}\right\|_\infty^2
       \left\|\sqrt{\rho_{AB}}\right\|_2^2
 \le\dim A.
\label{eq:conditional-power-first-moment-bound}
\end{equation}
This argument also applies when either marginal or the joint state is singular.
\end{proof}

Fix a block length $n$. For coherent information, let
$\omega_{RB^nE^n}$ be the Stinespring output of a pure input extension
$\psi_{RA^n}$. For $0<s<1$, let $\sigma_{B^n}^\star$ and
$\tau_{E^n}^\star$ be the optimizers from
Lemma~\ref{lem:tilted-state} and define
\begin{equation}
 \zeta_{B^nE^n}
 =\frac{((\sigma_{B^n}^\star)^{-s/2}\otimes(\tau_{E^n}^\star)^{s/2})
          \omega_{B^nE^n}
       ((\sigma_{B^n}^\star)^{-s/2}\otimes(\tau_{E^n}^\star)^{s/2})}
      {\Tr[\omega_{B^nE^n}((\sigma_{B^n}^\star)^{-s}
                           \otimes(\tau_{E^n}^\star)^{s})]}.
 \label{eq:block-quantum-tilted}
\end{equation}
Thus $\zeta_{B^n}=\sigma_{B^n}^\star$ and
$\zeta_{E^n}=\tau_{E^n}^\star$.

For Holevo information, choose a finite family of pure inputs and write
\begin{equation}
 \omega_{B^nE^n}^x=U^{\otimes n}\psi_{A^n}^x(U^{\otimes n})^\dagger,
 \qquad \omega_{B^n}^x=\mathcal N^{\otimes n}(\psi_{A^n}^x).
 \label{holevo:eq:stinespring-outputs}
\end{equation}
Let $p^\star$, $\sigma_{B^n}^\star$ and
$\{\tau_{E^n}^{x,\star}\}_x$ be the saddle data from
Lemma~\ref{holevo:lem:tilted-states}, and define
\begin{align}
 \xi_{B^nE^n}^x
 &=\frac{((\sigma_{B^n}^\star)^{-s/2}\otimes
          (\tau_{E^n}^{x,\star})^{s/2})
        \omega_{B^nE^n}^x
        ((\sigma_{B^n}^\star)^{-s/2}\otimes
          (\tau_{E^n}^{x,\star})^{s/2})}
       {\Tr[\omega_{B^nE^n}^x
         ((\sigma_{B^n}^\star)^{-s}\otimes
          (\tau_{E^n}^{x,\star})^s)]},
 \label{holevo:eq:block-tilted-conditional}\\
 \xi_{XB^nE^n}&=\sum_xp_x^\star\ket x\!\bra x\otimes\xi_{B^nE^n}^x,
 \qquad
 \tau_{XE^n}^\star=\sum_xp_x^\star\ket x\!\bra x
                      \otimes\tau_{E^n}^{x,\star}.
 \label{holevo:eq:block-tilted}
\end{align}
Thus $\xi_{B^n}=\sigma_{B^n}^\star$ and
$\xi_{XE^n}=\tau_{XE^n}^\star$.

For each $i$, let $P_i$ be $UU^\dagger$ on $B_iE_i$, tensored with the
identity on $B_{-i}E_{-i}$. The projections commute.
The symbols $B_{-i}$ and $E_{-i}$ denote all sites except $i$.

The main new step in the following lemma is to pull the exact
Stinespring-support identity through the tilted-state powers and
localize the resulting deviation to a single channel use. This makes
Lemma~\ref{lem:conditional-powers} applicable with all remaining channel
uses as conditioning systems, so that the support error is bounded
independently of $n$.

\begin{lemma}[Local support estimate]\label{lem:local-support}
Set
\begin{equation}
 \kappa:=\sqrt{\dim B-1}+(\dim B)^{1/4}\sqrt{\dim E-1}.
 \label{eq:local-support-kappa}
\end{equation}
For every $0<s\le1/2$ and every $i$, the coherent-information tilted
state in Eq.~\eqref{eq:block-quantum-tilted} satisfies
\begin{equation}
 \Tr[(I-P_i)\zeta_{B^nE^n}]\le2\kappa^2s^2.
 \label{eq:local-support-quantum}
\end{equation}
For the same range of $s$, the Holevo tilted state in
Eq.~\eqref{holevo:eq:block-tilted} satisfies
\begin{equation}
 \Tr[(I_X\otimes(I-P_i))\xi_{XB^nE^n}]\le2\kappa^2s^2.
 \label{holevo:eq:local-support}
\end{equation}
\end{lemma}

\begin{proof}
We first prove the coherent-information statement. The original Stinespring
output $\omega_{B^nE^n}$ is supported on $P_i$. Conjugating
Eq.~\eqref{eq:block-quantum-tilted} by
$(\sigma_{B^n}^\star)^{s/2}\otimes(\tau_{E^n}^\star)^{-s/2}$ therefore gives
\begin{equation}
 (I-P_i)((\sigma_{B^n}^\star)^{s/2}
          \otimes(\tau_{E^n}^\star)^{-s/2})
                    \sqrt{\zeta_{B^nE^n}}=0.
 \label{eq:inverse-support-quantum}
\end{equation}
Multiply Eq.~\eqref{eq:inverse-support-quantum} on the left by
$(\sigma_{B_{-i}}^\star)^{-s/2}
 \otimes(\tau_{E_{-i}}^\star)^{s/2}$, with identities
on $B_iE_i$. This operator commutes with $P_i$. For this calculation put
\begin{align}
 C_{B^n}&=(I_{B_i}\otimes\sigma_{B_{-i}}^\star)^{-s/2}
           (\sigma_{B^n}^\star)^{s/2},
 \label{eq:two-block-operator-B-quantum}\\
 C_{E^n}&=(I_{E_i}\otimes\tau_{E_{-i}}^\star)^{s/2}
          (\tau_{E^n}^\star)^{-s/2}.
 \label{eq:two-block-operators-quantum}
\end{align}
Then $(I-P_i)(C_{B^n}\otimes C_{E^n})\sqrt\zeta=0$. Consequently,
\begin{align}
 \sqrt{\Tr[(I-P_i)\zeta]}
 &\le\norm{(I-C_{B^n}\otimes C_{E^n})\sqrt\zeta}_2
 \label{eq:local-estimate-quantum-projection}\\
 &\le\norm{C_{B^n}}_\infty\norm{(C_{E^n}-I)\sqrt{\tau_{E^n}^\star}}_2
       +\norm{(C_{B^n}-I)\sqrt{\sigma_{B^n}^\star}}_2
 \label{eq:local-estimate-quantum-triangle}\\
 &\le\sqrt2\,s\left[(\dim B)^{s/2}\sqrt{\dim E-1}
                         +\sqrt{\dim B-1}\right].
 \label{eq:local-estimate-quantum}
\end{align}
The second line expands
$C_{B^n}\otimes C_{E^n}-I=C_{B^n}\otimes(C_{E^n}-I)+(C_{B^n}-I)\otimes I$
and uses the exact marginals of $\zeta$. The last line applies
Lemma~\ref{lem:conditional-powers} at parameters $s$ and $-s$, with systems
$B_i,B_{-i}$ and $E_i,E_{-i}$. Since $s\le1/2$, the expression in
parentheses is at most $\kappa$. Squaring proves
Eq.~\eqref{eq:local-support-quantum}.

For the Holevo state, each $\omega_{B^nE^n}^x$ is supported on $P_i$.
On the $x$ block, conjugating Eq.~\eqref{holevo:eq:block-tilted} by
$(\sigma_{B^n}^\star)^{s/2}\otimes(\tau_{XE^n}^\star)^{-s/2}$ gives
$2^{-s\RenyiHolevoQuantity{1/(1-s)}(\{\omega_{B^n}^x\}_x)}
(p_x^\star)^{1-s}\omega_{B^nE^n}^x$. Hence
\begin{equation}
 (I_X\otimes(I-P_i))((\sigma_{B^n}^\star)^{s/2}
                     \otimes(\tau_{XE^n}^\star)^{-s/2})
                    \sqrt{\xi_{XB^nE^n}}=0.
 \label{holevo:eq:inverse-support}
\end{equation}
Here and below tensor factors are placed on their labelled systems.
Define
\begin{align}
 D_{B^n}&=(I_{B_i}\otimes\sigma_{B_{-i}}^\star)^{-s/2}
          (\sigma_{B^n}^\star)^{s/2},
 \label{holevo:eq:two-block-operator-B}\\
 D_{XE^n}&=(I_{E_i}\otimes\tau_{XE_{-i}}^\star)^{s/2}
         (\tau_{XE^n}^\star)^{-s/2}.
 \label{holevo:eq:two-block-operators}
\end{align}
Repeating the preceding expansion, now using
$\xi_{B^n}=\sigma_{B^n}^\star$ and
$\xi_{XE^n}=\tau_{XE^n}^\star$, gives
\begin{align}
 \sqrt{\Tr[(I_X\otimes(I-P_i))\xi]}
 &\le\norm{D_{B^n}}_\infty\norm{(D_{XE^n}-I)\sqrt{\tau_{XE^n}^\star}}_2
       +\norm{(D_{B^n}-I)\sqrt{\sigma_{B^n}^\star}}_2
 \label{holevo:eq:local-estimate-triangle}\\
 &\le\sqrt2\,s\left[(\dim B)^{s/2}\sqrt{\dim E-1}
                         +\sqrt{\dim B-1}\right]
 \le\sqrt2\,\kappa s.
 \label{holevo:eq:local-estimate}
\end{align}
The two uses of Lemma~\ref{lem:conditional-powers} have local systems
$B_i$ and $E_i$ and conditioning systems $B_{-i}$ and $XE_{-i}$,
respectively. Squaring proves Eq.~\eqref{holevo:eq:local-support}.
\end{proof}

An $O(s)$ change in the state-weighted norm thus gives an $O(s^2)$
support error. The conditioning systems may contain all other channel
uses and the classical label $X$, yet only $\dim B_i$ and $\dim E_i$
enter the bound. No product-state assumption or bound on the alphabet
size is required, and no global trace-distance approximation is made.

\subsection{Restoring the Stinespring support}

The following argument adapts the Leung--Smith telescoping technique
for continuity of regularized channel capacities
\cite{leung2009continuity}: we restore one local Stinespring constraint
at a time and use conditional-entropy continuity to charge only the
dimension of the repaired system. The same restoration channel works
for both tasks and preserves the joint marginal of all unmodified
systems. Summing the local entropy costs gives a bound linear in the
block length.

Two external continuity estimates enter the argument. We state them in the
normalizations used below.

\begin{fact}[Fuchs--van de Graaf inequality]
\label{fact:fuchs-van-de-graaf}
For states $\rho$ and $\sigma$, let
$F(\rho,\sigma)=\norm{\sqrt\rho\sqrt\sigma}_1^2$ be their fidelity. Then
\begin{equation}
 \frac12\norm{\rho-\sigma}_1
 \le\sqrt{1-F(\rho,\sigma)}.
 \label{eq:fact-fuchs-van-de-graaf}
\end{equation}
This is the right-hand inequality in
Ref.~\cite[Theorem~1]{fuchs1999cryptographic}.
\end{fact}

\begin{fact}[Conditional-entropy continuity]
\label{fact:conditional-entropy-continuity}
Let $\rho_{AC}$ and $\sigma_{AC}$ be states, let $d_A=\dim A\ge2$, and suppose
\begin{equation}
 \frac12\norm{\rho_{AC}-\sigma_{AC}}_1\le\delta
 \quad\text{with}\quad 0\le\delta\le1-d_A^{-2}.
\end{equation}
Then, in base-two units,
\begin{equation}
 \left|H(A|C)_\rho-H(A|C)_\sigma\right|
 \le\delta\log(d_A^2-1)+h_2(\delta).
 \label{eq:fact-conditional-entropy-continuity}
\end{equation}
This is the sharp conditional-entropy continuity bound of
Ref.~\cite[Theorem~1.1]{berta2026sharpcontinuity}; it also follows by taking
the order-one limit of the R\'enyi continuity bounds in
Ref.~\cite[Eq.~(20)]{chengliu2026sharp}.
\end{fact}

\begin{lemma}[Support restoration]\label{lem:restoration}
Fix $n\ge1$. Let $\dim B,\dim E\ge2$, and let $\vartheta_{XB^nE^n}$ be a state classical on
$X$, which may be one dimensional. Suppose
$\Tr[(I-P_i)\vartheta_{XB^nE^n}]\le\delta^2/2$ for every $i$, with $0\le\delta\le1/2$.
Then
\begin{align}
 H(B^n)_\vartheta-H(E^n)_\vartheta
 &\le\CoherentInformation(\mathcal N^{\otimes n})\nonumber\\
 &\quad+n\delta\log\!\left[((\dim B)^2-1)((\dim E)^2-1)\right]
       +2n h_2(\delta),
 \label{eq:entropy-restoration}\\
 I(X:B^n)_\vartheta
 &\le\HolevoQuantity(\mathcal N^{\otimes n})
       +2n\delta\log\!\left((\dim B)^2-1\right)+2n h_2(\delta).
 \label{holevo:eq:entropy-restoration}
\end{align}
\end{lemma}

\begin{proof}
Choose a fixed density operator $\eta_{BE}$ supported on
$\operatorname{im}U$ and define the channel
\begin{equation}
 \mathcal T_{BE}(Z_{BE})=UU^\dagger Z_{BE}UU^\dagger
       +\Tr[(I-UU^\dagger)Z_{BE}]\eta_{BE}.
 \label{eq:restoring-channel}
\end{equation}
It retains the allowed component and replaces the remaining component
by $\eta_{BE}$. Both terms are completely positive, their traces add
to $\Tr[Z_{BE}]$, and every output lies in $\operatorname{im}U$.
Apply its copy on each pair $B_iE_i$ successively:
\begin{align}
 \vartheta^{(0)}_{XB^nE^n}&=\vartheta_{XB^nE^n},
 \label{eq:restoration-sequence-initial}\\
 \vartheta^{(i)}_{XB^nE^n}
 &=\bigl(\mathcal T_{B_iE_i}\otimes\id_{XB_{-i}E_{-i}}\bigr)
       (\vartheta^{(i-1)}_{XB^nE^n}),\qquad i=1,\ldots,n.
 \label{eq:restoration-sequence}
\end{align}
The superscript $(i)$ counts completed repairs; all these states remain
on $XB^nE^n$. After step $i$, the state is supported on $P_i$.
Every later channel acts on a different pair and hence preserves this
constraint. Since the $P_i$ commute and their product is
$I_X\otimes(UU^\dagger)^{\otimes n}$, the final state is supported on
$X\otimes\operatorname{im}(U^{\otimes n})$. All steps act trivially on
$X$, so the states remain classical on $X$ and have marginal $\vartheta_X$.

Earlier operations leave the marginal on $B_iE_i$ unchanged, so
\begin{equation}
 \Tr[(I-P_i)\vartheta^{(i-1)}_{XB^nE^n}]
 =\Tr[(I-P_i)\vartheta_{XB^nE^n}].
 \label{eq:unchanged-local-probability}
\end{equation}
The extended repair channel has a Kraus representation containing $P_i$.
Apply its Stinespring isometry to a purification of
$\vartheta^{(i-1)}_{XB^nE^n}$ and compare with that same purification tensored
with the environment flag for $P_i$. Their overlap is
$\Tr[P_i\vartheta^{(i-1)}_{XB^nE^n}]$; hence
\begin{equation}
 F\!\left(\vartheta^{(i)}_{XB^nE^n},
          \vartheta^{(i-1)}_{XB^nE^n}\right)
 \ge\bigl(\Tr[P_i\vartheta^{(i-1)}_{XB^nE^n}]\bigr)^2.
\label{eq:restoration-fidelity}
\end{equation}
Fact~\ref{fact:fuchs-van-de-graaf} yields
\begin{align}
 \frac12\norm{\vartheta^{(i)}_{XB^nE^n}-\vartheta^{(i-1)}_{XB^nE^n}}_1
 &\le\sqrt{1-\bigl(\Tr[P_i\vartheta^{(i-1)}_{XB^nE^n}]\bigr)^2}
 \label{eq:one-step-distance-fidelity}\\
 &\le\sqrt{2\Tr[(I-P_i)\vartheta_{XB^nE^n}]}\le\delta.
 \label{eq:one-step-distance}
\end{align}
Because each local channel is trace preserving, it leaves the joint
marginal of the untouched systems unchanged:
\begin{equation}
 \vartheta^{(i)}_{XB_{-i}E_{-i}}=\vartheta^{(i-1)}_{XB_{-i}E_{-i}}.
\label{eq:restoration-preserves-complement}
\end{equation}

For the first claim, the entropy chain rule consequently gives
\begin{equation}
 H(B^n)_{\vartheta^{(i)}}-H(B^n)_{\vartheta^{(i-1)}}
 =H(B_i|B_{-i})_{\vartheta^{(i)}}-H(B_i|B_{-i})_{\vartheta^{(i-1)}}.
 \label{eq:conditional-difference}
\end{equation}
Partial trace in \eqref{eq:one-step-distance} shows that
$\vartheta^{(i)}_{B_iB_{-i}}$ and $\vartheta^{(i-1)}_{B_iB_{-i}}$ are within
trace distance $\delta$. Apply the sharp conditional-entropy continuity
bound in Fact~\ref{fact:conditional-entropy-continuity} with first system $B_i$
and conditioning system $B_{-i}$. In base-two units,
\eqref{eq:conditional-difference} gives
\begin{equation}
 \left|H(B^n)_{\vartheta^{(i)}}-H(B^n)_{\vartheta^{(i-1)}}\right|
 \le\delta\log((\dim B)^2-1)+h_2(\delta).
 \label{eq:output-entropy-change}
\end{equation}
The dimension factor is $(\dim B_i)^2-1=(\dim B)^2-1$; the dimension
of $B_{-i}$ does not enter. Applying the same argument to
$\vartheta^{(i)}_{E_iE_{-i}}$ and $\vartheta^{(i-1)}_{E_iE_{-i}}$ gives
\begin{equation}
 \left|H(E^n)_{\vartheta^{(i)}}-H(E^n)_{\vartheta^{(i-1)}}\right|
 \le\delta\log((\dim E)^2-1)+h_2(\delta).
\label{eq:environment-entropy-change}
\end{equation}
These applications are valid because $\delta\le1/2$ lies below both
thresholds $1-(\dim B)^{-2}$ and $1-(\dim E)^{-2}$.
Adding the two estimates and telescoping from $\vartheta^{(0)}=\vartheta$ to
$\vartheta^{(n)}$ bounds the difference between
$H(B^n)_\vartheta-H(E^n)_\vartheta$ and
$H(B^n)_{\vartheta^{(n)}}-H(E^n)_{\vartheta^{(n)}}$ by the correction in
\eqref{eq:entropy-restoration}.
The final $B^nE^n$ marginal is the Stinespring output of the input
$(U^{\otimes n})^\dagger\vartheta^{(n)}_{B^nE^n}U^{\otimes n}$, a positive
operator of trace one. Consequently,
the definition of coherent information in
Eq.~\eqref{eq:intro-quantum-capacity} gives
\begin{equation}
 H(B^n)_{\vartheta^{(n)}}-H(E^n)_{\vartheta^{(n)}}
 \le\CoherentInformation(\mathcal N^{\otimes n}),
\label{eq:repaired-coherent-information-bound}
\end{equation}
which proves the first claim.

For the second claim, the equality
$\vartheta^{(i)}_{XB_{-i}}=\vartheta^{(i-1)}_{XB_{-i}}$ and the mutual-information
chain rule give the exact identity
\begin{align}
 I(X:B^n)_{\vartheta^{(i)}}-I(X:B^n)_{\vartheta^{(i-1)}}
 &=\bigl[H(B_i|B_{-i})_{\vartheta^{(i)}}
          -H(B_i|B_{-i})_{\vartheta^{(i-1)}}\bigr]\nonumber\\
 &\quad-\bigl[H(B_i|XB_{-i})_{\vartheta^{(i)}}
          -H(B_i|XB_{-i})_{\vartheta^{(i-1)}}\bigr].
 \label{holevo:eq:local-entropy-difference}
\end{align}
Apply the continuity bound to the pairs of reduced states
$(\vartheta^{(i)}_{B_iB_{-i}},\vartheta^{(i-1)}_{B_iB_{-i}})$ and
$(\vartheta^{(i)}_{B_iXB_{-i}},\vartheta^{(i-1)}_{B_iXB_{-i}})$.
Their trace distances are at most $\delta$ by
\eqref{eq:one-step-distance}, and the first system is $B_i$ in both
applications. The triangle inequality therefore gives
\begin{equation}
 \left|I(X:B^n)_{\vartheta^{(i)}}-I(X:B^n)_{\vartheta^{(i-1)}}\right|
 \le2\delta\log((\dim B)^2-1)+2h_2(\delta).
\label{eq:mutual-information-step-bound}
\end{equation}
The operator
$(I_X\otimes(U^{\otimes n})^\dagger)\vartheta^{(n)}_{XB^nE^n}
 (I_X\otimes U^{\otimes n})$
is a density operator on $XA^n$, classical on $X$, whose channel output
is $\vartheta^{(n)}_{XB^n}$. It thus specifies an admissible input ensemble, so
\begin{equation}
 I(X:B^n)_{\vartheta^{(n)}}\le\HolevoQuantity(\mathcal N^{\otimes n}).
\label{eq:repaired-holevo-bound}
\end{equation}
Telescoping the preceding estimate from $\vartheta^{(0)}=\vartheta$ to
$\vartheta^{(n)}$ proves \eqref{holevo:eq:entropy-restoration}.
\end{proof}

The conditional states of $\vartheta^{(i)}_{XB^nE^n}$ can become mixed during
restoration. Purity is used only in the preceding Holevo derivative
identity; restoration compares $I(X:B^n)_{\vartheta^{(i)}}$ for
$i=0,\ldots,n$. These channels construct
admissible inputs for the two optimizations and need not be implemented
by a receiver. The entropy estimates use the distance at each step,
so the initial and final states need not be globally close.

\section{Continuity of the R\'enyi capacities}
\label{sec:continuity}

We now combine the integral representations of
Section~\ref{sec:integral-representations} with the projection estimates of
Section~\ref{sec:projection}. With $\kappa$ as defined in
Eq.~\eqref{eq:local-support-kappa}, put $s_\alpha=(\alpha-1)/\alpha\in(0,1)$.
For $\dim B,\dim E\ge2$ and $2\kappa s_\alpha\le1/2$, define
\begin{align}
 \varepsilon_{\rm q}(\alpha)
 &=2\kappa s_\alpha\log\!\left[((\dim B)^2-1)((\dim E)^2-1)\right]
      +2h_2(2\kappa s_\alpha),
 \label{eq:continuity-modulus-quantum}\\
 \varepsilon_{\rm c}(\alpha)
 &=4\kappa s_\alpha\log\!\left((\dim B)^2-1\right)
      +2h_2(2\kappa s_\alpha).
 \label{eq:continuity-moduli}
\end{align}

\begin{proposition}[Finite-block continuity bounds]
\label{prop:finite-block-continuity}
Let $\dim B,\dim E\ge2$ and $\alpha>1$ satisfy
$\alpha\le4\kappa/(4\kappa-1)$. Then, for every $n\ge1$,
\begin{align}
 0&\le\RenyiCoherentInformation{\alpha}(\mathcal N^{\otimes n})
       -\CoherentInformation(\mathcal N^{\otimes n})
       \le n\varepsilon_{\rm q}(\alpha),
 \label{eq:finite-block-quantum}\\
 0&\le\RenyiHolevoQuantity{\alpha}(\mathcal N^{\otimes n})
       -\HolevoQuantity(\mathcal N^{\otimes n})
       \le n\varepsilon_{\rm c}(\alpha).
 \label{eq:finite-block-classical}
\end{align}
Consequently, $0\le\RenyiQuantumCapacity{\alpha}(\mathcal N)
-\QuantumCapacity(\mathcal N)\le\varepsilon_{\rm q}(\alpha)$ and
$0\le\RenyiClassicalCapacity{\alpha}(\mathcal N)
-\ClassicalCapacity(\mathcal N)\le\varepsilon_{\rm c}(\alpha)$.
\end{proposition}

\begin{proof}
For the quantum bound, we first prove
Eq.~\eqref{eq:finite-block-quantum}. For the matrix-power
argument, restrict $B$ and $E$ to the supports of $\mathcal N(I_A)$ and
$\mathcal N^{\mathrm c}(I_A)$; the corresponding operators are then strictly
positive. Keeping the original dimensions in the bounds can only enlarge
the right-hand side. Fix $n$ and a positive definite input $\rho_{A^n}$,
and let $\omega_{RB^nE^n}$ be the Stinespring output of a purification.
Its two marginals are strictly positive because, for example,
\begin{equation}
 \mathcal N^{\otimes n}(\rho_{A^n})
 \ge\lambda_{\min}(\rho_{A^n})\mathcal N(I_A)^{\otimes n}>0,
\label{eq:positive-output-from-faithful-input}
\end{equation}
and likewise for the complementary channel.

For every $0<u\le s_\alpha$, Lemma~\ref{lem:local-support} gives, at every site,
\begin{equation}
 \Tr[(I-P_i)\zeta_{B^nE^n}(u)]
 \le2u^2\kappa^2\le\frac{(2\kappa s_\alpha)^2}{2}.
 \label{eq:all-smaller-orders}
\end{equation}
The assumption on $\alpha$ is equivalent to $2\kappa s_\alpha\le1/2$;
it also ensures $s_\alpha\le1/2$, as required in
Lemma~\ref{lem:local-support}. Apply Lemma~\ref{lem:restoration} with
$\delta=2\kappa s_\alpha$. Combining this with
Lemma~\ref{lem:derivative} therefore
gives, for $0<u<s_\alpha$,
\begin{equation}
 \frac{\mathrm{d}}{\mathrm{d}u}\left[-u\widetilde H_{1/(1-u)}^\uparrow(R|B^n)_\omega\right]
 \le\CoherentInformation(\mathcal N^{\otimes n})
      +n\varepsilon_{\rm q}(\alpha).
 \label{eq:uniform-derivative}
\end{equation}
The differentiated function is continuously differentiable on
$[0,s_\alpha]$ and vanishes at zero. Integration and division by
$s_\alpha$ yield
\begin{equation}
 -\widetilde H_\alpha^\uparrow(R|B^n)_\omega
 \le\CoherentInformation(\mathcal N^{\otimes n})
      +n\varepsilon_{\rm q}(\alpha).
 \label{eq:input-bound}
\end{equation}
The repaired state may correspond to a different channel input, which is
allowed in the coherent-information optimization.

Approximate a singular input by
$(1-t)\rho_{A^n}+tI_{A^n}/\dim A^n$ and let $t\searrow0$.
Choose canonical purifications on a fixed reference system; these
purifications, and hence their channel outputs, converge in trace norm.
Fact~\ref{fact:conditional-renyi-continuity} below therefore extends
Eq.~\eqref{eq:input-bound} to every input. Taking the supremum and using the order monotonicity in
Fact~\ref{fact:sandwiched-data-processing} proves
Eq.~\eqref{eq:finite-block-quantum}.

For the classical bound, fix a finite family of pure inputs on $A^n$. The
second statement of Lemma~\ref{lem:local-support} gives
\begin{equation}
 \Tr[(I_X\otimes(I-P_i))\xi_{XB^nE^n}(u)]
 \le2u^2\kappa^2\le\frac{(2\kappa s_\alpha)^2}{2},
 \label{holevo:eq:all-smaller-orders}
\end{equation}
uniformly in the family and its size. Hence
Lemmas~\ref{holevo:lem:derivative} and~\ref{lem:restoration} give, for
almost every $0<u<s_\alpha$,
\begin{equation}
 \frac{\mathrm{d}}{\mathrm{d}u}\left[
 u\RenyiHolevoQuantity{1/(1-u)}(\{\omega_{B^n}^x\}_x)\right]
 \le\HolevoQuantity(\mathcal N^{\otimes n})
      +n\varepsilon_{\rm c}(\alpha).
 \label{holevo:eq:uniform-derivative}
\end{equation}
The function on the left is absolutely continuous and vanishes at zero.
Integrating gives the corresponding bound for the fixed family at order
$\alpha$. Every finite mixed-input ensemble has a finite pure-state
refinement, and forgetting the refinement is a channel, so data processing
shows that restricting to pure input families does not lower the supremum.
Taking that supremum and again using the order monotonicity in
Fact~\ref{fact:sandwiched-data-processing} proves
Eq.~\eqref{eq:finite-block-classical}.

Dividing Eqs.~\eqref{eq:finite-block-quantum} and
\eqref{eq:finite-block-classical} by $n$ and taking the regularized limits
proves the asserted regularized bounds; existence of these limits was already
established in Section~\ref{sec:arimoto}.
\end{proof}

\begin{fact}[Fixed-order continuity of the conditional R\'enyi entropy]
\label{fact:conditional-renyi-continuity}
Let $\alpha>1$, put $\alpha'=\alpha/(\alpha-1)$, and let
$\rho_{AB},\sigma_{AB}$ be states satisfying
$\tfrac12\|\rho_{AB}-\sigma_{AB}\|_1\le\epsilon$. Then
\begin{equation}
 \left|\widetilde H_\alpha^\uparrow(A|B)_\rho
       -\widetilde H_\alpha^\uparrow(A|B)_\sigma\right|
 \le \alpha'\log\!\left(1+2\epsilon(\dim A)^{2/\alpha'}\right).
\label{eq:fixed-order-conditional-renyi-continuity}
\end{equation}
This is Ref.~\cite[Theorem~5.2]{BG} in the notation of this paper. In
particular, at fixed $\alpha$ the optimized conditional R\'enyi entropy is
continuous in the state.
\end{fact}

\begin{proof}[Proof of Theorem~\ref{thm:main}]
Since $2\kappa s_\alpha\to0$ and $h_2(2\kappa s_\alpha)\to0$ as
$\alpha\searrow1$, Proposition~\ref{prop:finite-block-continuity} proves both
limits when $\dim B,\dim E\ge2$. If $\dim B=1$, the coherent and Holevo
quantities both vanish. If $\dim E=1$, the channel is isometric and all four
finite-block quantities in Eqs.~\eqref{eq:finite-block-quantum} and
\eqref{eq:finite-block-classical} equal $n\log\dim A$. Thus the differences
also vanish in the cases excluded from
Proposition~\ref{prop:finite-block-continuity}.
\end{proof}

\section{Conclusion}
\label{sec:conclusion}

We have established exponential strong converses for unassisted quantum and
classical communication over every finite-dimensional memoryless quantum
channel. Our main technical results are one-shot integral representations for
the R\'enyi coherent and Holevo information and asymptotic continuity of the
corresponding R\'enyi capacities at order one. Local support estimates and
sequential restoration of the Stinespring image provide the
block-length-independent per-use control needed to combine these results with the Arimoto converse
bounds. This yields exact regularized strong-converse exponent formulas and,
crucially, proves that the exponents are strictly positive at every rate above
the quantum and classical capacities. The exact quantum exponent concerns
entanglement generation; the exponential strong converse extends to the other
standard unassisted quantum communication criteria, although their exact
exponents can differ.

Our uniform bounds also extend to nonstationary product channels $\bigotimes_{i=1}\mathcal N_i$ with uniformly bounded local output and Stinespring environment dimensions. The local support estimates and sequential restoration depend only on these dimensions and the tensor-product dilation, allowing different component channels and arbitrary correlations among their inputs. Consequently, for every $n$ and $\alpha>1$ sufficiently close to one, independently of $n$, \begin{align}
0&\le Q_\alpha^{(1)}\!\left(\bigotimes_{i=1}^n\mathcal N_i\right)
-Q^{(1)}\!\left(\bigotimes_{i=1}^n\mathcal N_i\right)
\le \sum_{i=1}^n\varepsilon_{{\rm q},i}(\alpha),
\label{eq:nonidentical-quantum}\\
0&\le C_\alpha^{(1)}\!\left(\bigotimes_{i=1}^n\mathcal N_i\right)
-C^{(1)}\!\left(\bigotimes_{i=1}^n\mathcal N_i\right)
\le \sum_{i=1}^n\varepsilon_{{\rm c},i}(\alpha),
\label{eq:nonidentical-classical}
\end{align}
where the correction terms in Eqs.~\eqref{eq:nonidentical-quantum}
and~\eqref{eq:nonidentical-classical} are obtained from
$\varepsilon_{\rm q}(\alpha)$ in \eqref{eq:continuity-modulus-quantum} and $\varepsilon_{\rm c}(\alpha)$ in \eqref{eq:continuity-moduli},
respectively, by substituting the local dimensions
$\dim B_i,\dim E_i$ at site $i$.
Uniform boundedness of these dimensions ensures that the corrections
per channel use vanish uniformly in $n$ as $\alpha\searrow1$. 
For channel sequences whose normalized coherent and Holevo information converge to the respective operational capacities, these bounds also establish exponential strong converses.

\bigskip
\paragraph*{Acknowledgements}
H.-C.C. acknowledges support from the National Science and Technology Council
(NSTC 115-2628-E-002-005, NSTC 114-2119-M-001-002, and NSTC
115-2124-M-002-014) and the Ministry of Education (NTU-115V2016-1,
NTU-CC115L893705, and NTU-115L900702). M.T. acknowledges support from the
National Research Foundation Investigatorship Award
(NRF-NRFI10-2024-0006) and the National Research Foundation, Singapore,
through the National Quantum Office, hosted by A*STAR, under its Centre for
Quantum Technologies Funding Initiative (S24Q2d0009).

\appendix

\section{Quantum Arimoto converse}
\label{app:one-shot}

\begin{proof}[Proof of Proposition~\ref{prop:arimoto-quantum}]
For a fixed decoder, the generation fidelity is linear in the preparation
$\rho_{RA^n}$. It is therefore an average of the fidelities obtained from
the pure states in any decomposition of $\rho_{RA^n}$. Since the right
side of \eqref{eq:code-bound} is independent of the preparation, it
suffices to prove the bound for a pure state $\psi_{RA^n}$.
Let $\omega_{RB^n}=(\id_R\otimes\mathcal N^{\otimes n})(\psi_{RA^n})$
be the resulting state immediately before decoding.
For any density operator $\sigma_{B^n}$, apply the decoder followed by
the binary test $\{\Phi_{RB_0},I-\Phi_{RB_0}\}$ to
$\omega_{RB^n}$ and $I_R\otimes\sigma_{B^n}$.
The comparison operator has weight
\begin{equation}
 \Tr\!\left[\Phi_{RB_0}(I_R\otimes\mathcal D(\sigma_{B^n}))\right]=\frac1M
\end{equation}
on the first outcome.
Data processing for sandwiched divergence \cite{beigi2013data}, followed by
retaining this one nonnegative term in the classical divergence, gives
\begin{align}
 \widetilde D_\alpha(\omega_{RB^n}\|I_R\otimes\sigma_{B^n})
 &\ge\frac{1}{\alpha-1}\log
             (F_{\mathrm{eg}}^{\alpha}M^{\alpha-1})
 \label{eglower:eq:binary-test-bound}\\
 &=\log M+\frac{\alpha}{\alpha-1}\log F_{\mathrm{eg}}.
\end{align}
Taking the infimum over $\sigma_{B^n}$ yields
\begin{equation}
 -\widetilde H_\alpha^\uparrow(R|B^n)_\omega
 \ge\log M+\frac{\alpha}{\alpha-1}\log F_{\mathrm{eg}}.
 \label{eq:bell-test}
\end{equation}

The preparation $\psi_{RA^n}$ is pure and has Schmidt rank at most
$\dim A^n$. Reference isometries leave the conditional entropy unchanged,
so the optimization in \eqref{eq:coherent-information} gives
\begin{equation}
 -\widetilde H_\alpha^\uparrow(R|B^n)_\omega
 \le\RenyiCoherentInformation{\alpha}(\mathcal N^{\otimes n}).
 \label{eq:generation-input-bound}
\end{equation}
Combining \eqref{eq:generation-input-bound} with \eqref{eq:bell-test}
proves \eqref{eq:code-bound}, including mixed preparations by the
linearity argument above. The bound is
also immediate when $F_{\mathrm{eg}}=0$.

It remains to pass from the one-shot quantity to its regularization.
The supremum representation in Eq.~\eqref{eq:regularized-supremum} gives
\begin{equation}
 \frac1n\RenyiCoherentInformation{\alpha}(\mathcal N^{\otimes n})
 \le \RenyiQuantumCapacity{\alpha}(\mathcal N)
 \qquad(n\ge1).
\end{equation}
Writing $r=n^{-1}\log M$ and substituting the preceding inequality into
Eq.~\eqref{eq:code-bound} proves
Proposition~\ref{prop:arimoto-quantum}, including
Eq.~\eqref{eq:regularized-arimoto-bound}.
\end{proof}

\section{Classical Arimoto converse}
\label{app:classical-operational}

\begin{proof}[Proof of Proposition~\ref{prop:arimoto-classical}]
For an $(n,M)$ classical code, let
$\omega_{XB^n}=M^{-1}\sum_m\ket m\!\bra m\otimes
\mathcal N^{\otimes n}(\rho_{A^n}^m)$.
For any comparison state $\sigma_{B^n}$, the binary test with accepting
operator $\sum_m\ket m\!\bra m\otimes D_m$ has probabilities
$P_{\mathrm s}$ on $\omega_{XB^n}$ and $1/M$ on
$\omega_X\otimes\sigma_{B^n}$. The same data-processing argument as
in Appendix~\ref{app:one-shot} gives
\begin{equation}
 \widetilde D_\alpha(\omega_{XB^n}\|\omega_X\otimes\sigma_{B^n})
 \ge\log M+\frac{\alpha}{\alpha-1}\log P_{\mathrm s}.
\end{equation}
Taking the infimum over $\sigma_{B^n}$ and using
\eqref{holevo:eq:definition} gives
\begin{equation}
 P_{\mathrm s}\le
 2^{-\frac{\alpha-1}{\alpha}
       [\log M-\RenyiHolevoQuantity{\alpha}(\mathcal N^{\otimes n})]}.
\label{holevo:eq:one-shot-arimoto}
\end{equation}
The case $P_{\mathrm s}=0$ is immediate.

Additivity of the optimized divergence with fixed classical marginal
\cite[Lemma~7]{hayashitomamichel15c} and product input ensembles give
\begin{equation}
 \RenyiHolevoQuantity{\alpha}(\mathcal N^{\otimes(n+m)})
 \ge\RenyiHolevoQuantity{\alpha}(\mathcal N^{\otimes n})
      +\RenyiHolevoQuantity{\alpha}(\mathcal N^{\otimes m}).
 \label{holevo:eq:superadditivity}
\end{equation}
Thus Fekete's lemma shows that the limit in \eqref{holevo:eq:definition} exists and
\begin{equation}
 \RenyiClassicalCapacity{\alpha}(\mathcal N)
 =\sup_{n\ge1}\frac1n
       \RenyiHolevoQuantity{\alpha}(\mathcal N^{\otimes n}),
 \label{holevo:eq:regularized-supremum}
\end{equation}
which yields
\begin{equation}
 \frac1n\RenyiHolevoQuantity{\alpha}(\mathcal N^{\otimes n})
 \le\RenyiClassicalCapacity{\alpha}(\mathcal N)
 \qquad(n\ge1).
\end{equation}
Writing $r=n^{-1}\log M$ and combining the preceding two inequalities proves
Proposition~\ref{prop:arimoto-classical}.
\end{proof}

\section{Entanglement-generation strong-converse exponent}
\label{app:quantum-exponent}

We adapt the change-of-measure and pinching strategy used for
strong-converse exponent achievability in
Refs.~\cite{mosonyi2017strong,li2023strong}. For entanglement generation,
we use one-way hashing~\cite{devetakwinter2005distillation} or state
merging~\cite{berta2026tight}, together with the removal of forward
classical assistance by selecting a suitable preparation branch. We then
convexify the resulting relative-entropy bound and use universal pinching
to recover the sandwiched R\'enyi coherent information.

\subsection{Converse bound and positivity}

For each fixed $\alpha>1$, Eq.~\eqref{eq:code-bound} yields
\begin{equation}
 -\frac1n\log F_{\mathrm{eg}}^{\star}(n,r)
 \ge\frac{\alpha-1}{\alpha}
        \left[r-\frac1n\RenyiCoherentInformation{\alpha}(\mathcal N^{\otimes n})\right].
\end{equation}
Taking $n\to\infty$ and then the supremum over $\alpha>1$ proves
the lower bound on $E_{\mathrm{sc,eg}}(r,\mathcal N)$ in
Eq.~\eqref{eq:exponent-bound}. If $r>\QuantumCapacity(\mathcal N)$, choose
$0<\varepsilon<r-\QuantumCapacity(\mathcal N)$. By
Proposition~\ref{prop:finite-block-continuity}, one may fix $\alpha>1$ close
enough to one that
\begin{equation}
 \RenyiCoherentInformation{\alpha}(\mathcal N^{\otimes n})
 \le n[\QuantumCapacity(\mathcal N)+\varepsilon]\qquad(n\ge1).
\end{equation}
Substitution into \eqref{eq:code-bound} gives exponential decay and
\begin{equation}
 E_{\mathrm{sc,eg}}(r,\mathcal N)
 \ge\frac{\alpha-1}{\alpha}
       [r-\QuantumCapacity(\mathcal N)-\varepsilon]>0.
\end{equation}
The same bound applies to every code sequence whose limiting inferior rate
exceeds capacity, after choosing $r$ strictly between the two rates.

\subsection{Matching upper bound for entanglement generation}
\label{app:eg-upper}

The reverse inequality follows from three ingredients that we state
separately: a change-of-measure bound, its convexification, and a universal
pinching estimate.

\begin{lemma}[Change of measure for entanglement generation]
\label{egupper:lem:change-measure}
Fix $R\simeq A$. Let $\omega_{RB}$ be obtainable from one use of
$\mathcal N$ followed by a local channel on $B$. For every $r>0$ and every
state $\tau_{RB}$ with $\supp\tau\subseteq\supp\omega$,
\begin{equation}
 E_{\mathrm{sc,eg}}(r,\mathcal N)
 \le D(\tau\|\omega)+\max\{r+H(R|B)_\tau,0\},
 \label{egupper:eq:relative}
\end{equation}
where $D(\tau\|\omega)=\Tr[\tau(\log\tau-\log\omega)]$.
\end{lemma}

\begin{proof}
Suppose first that $H(R|B)_\tau<0$. One-way hashing
\cite{devetakwinter2005distillation} distills entanglement from
$\tau^{\otimes\ell}$ at every rate $0<t<-H(R|B)_\tau$. Equivalently,
purify $\tau_{RB}$ to $\tau_{RBE}$ and apply the state-merging bound of
\cite[Theorem~10]{berta2026tight} with Alice's system $R$, Bob's side
information $B$, and reference $E$; purity gives
$H(R|E)_\tau=-H(R|B)_\tau>0$. Discarding the transferred state leaves the
distilled entanglement.

Choose also $t<r$, put $K_\ell=2^{\lfloor\ell t\rfloor}$, and let
$p_\ell\to1$ be the squared fidelity of such a protocol on
$\tau^{\otimes\ell}$. Write $q_\ell$ for its squared fidelity on
$\omega^{\otimes\ell}$. Relative-entropy data processing through the protocol
and the target-projector test gives
\begin{equation}
 \ell D(\tau\|\omega)
 \ge-h_2(p_\ell)-p_\ell\log q_\ell.
 \label{egupper:eq:comparison}
\end{equation}
Only after this comparison do we remove the protocol's forward classical
assistance. On $\omega^{\otimes\ell}$, the squared fidelity $q_\ell$ is the
average of the conditional squared fidelities over Alice's instrument
outcomes, because the target is pure. Hence some outcome of positive
probability has conditional squared fidelity at least $q_\ell$. As in
Ref.~\cite[Lemma~22]{berta2026tight} and
Ref.~\cite[Lemma~14.6]{khatriwilde2024principles}, Alice's instrument acts
only on her retained system and commutes with the channel. Its selected branch can
therefore be absorbed into the input preparation: prepare the normalized
conditional input state deterministically, and fix Bob's decoder to the one
corresponding to that outcome, including the local channel used to obtain
$\omega$. The outcome is fixed in the code design, so no classical message is
needed. This is an admissible entanglement-generation code because the
retained marginal is unrestricted. Its squared fidelity is at least
$q_\ell$; preservation of $p_\ell$ is no longer needed.

Embedding the rank-$K_\ell$ target into rank
$M_\ell=\lceil2^{\ell r}\rceil$ multiplies its squared fidelity by
$K_\ell/M_\ell$. Letting $\ell\to\infty$ in
Eq.~\eqref{egupper:eq:comparison}, and then taking
$t\uparrow\min\{r,-H(R|B)_\tau\}$, proves
Eq.~\eqref{egupper:eq:relative}. If $H(R|B)_\tau\ge0$, product preparation
has fidelity $1/M_\ell$ and gives the stated, weaker upper bound directly.
\end{proof}

\begin{lemma}[Convexified relative-entropy bound]
\label{egupper:lem:convexification}
With the state domains of Lemma~\ref{egupper:lem:change-measure},
\begin{equation}
 E_{\mathrm{sc,eg}}(r,\mathcal N)
 \le\sup_{0<s<1}\left\{sr-
      \sup_{\omega,\tau}
      [-sH(R|B)_\tau-D(\tau\|\omega)]\right\}.
 \label{egupper:eq:time-sharing}
\end{equation}
\end{lemma}

\begin{proof}
Let $\mathcal C_{\mathcal N}$ be the closed convex hull of all finite pairs
$(d,h)=(D(\tau\|\omega),H(R|B)_\tau)$ allowed in
Lemma~\ref{egupper:lem:change-measure}. Applying that lemma to tensor products
of finitely many such pairs in rational proportions, and then approximating
real proportions, gives
\begin{equation}
 E_{\mathrm{sc,eg}}(r,\mathcal N)
 \le\inf_{(d,h)\in\mathcal C_{\mathcal N}}
       \bigl[d+\max\{r+h,0\}\bigr].
 \label{egupper:eq:convex-hull}
\end{equation}
Now use $\max\{x,0\}=\max_{0\le s\le1}sx$. For
$f(s;d,h)=d+s(r+h)$, Sion's theorem~\cite{sion1958general} gives the
oriented exchange
\begin{equation}
 \inf_{(d,h)\in\mathcal C_{\mathcal N}}\sup_{0\le s\le1}f(s;d,h)
 =\sup_{0\le s\le1}\inf_{(d,h)\in\mathcal C_{\mathcal N}}f(s;d,h).
\end{equation}
Indeed, $[0,1]$ is compact and convex, $\mathcal C_{\mathcal N}$ is convex,
and the payoff is finite, affine, and continuous on its domain. The infimum
of the linear functional $d+sh$ over the convex hull equals its infimum over
the original pairs. At $s=0$ this infimum is zero, since one may take
$\tau=\omega$. Moreover, $|h|\le\log\dim R$, so the resulting function of
$s$ is Lipschitz. Omitting the endpoints therefore does not change its
supremum, and Eq.~\eqref{egupper:eq:time-sharing} follows.
\end{proof}

\begin{lemma}[Universal pinching estimate]
\label{egupper:lem:pinching}
Choose a full-rank universal symmetric state $\sigma^u_{B^m}$ and let
$\mathcal P_m$ be its pinching map. There is a polynomial $v_m$, depending
only on $\dim B$, which bounds the number of distinct eigenvalues of
$\sigma^u_{B^m}$ and satisfies
$\sigma_{B^m}\le v_m\sigma^u_{B^m}$ for every permutation-invariant state
$\sigma_{B^m}$~\cite[Lemma~1]{li2023strong}. For any state $\omega_{RB}$, set
\begin{equation}
 \widehat\omega_{R^mB^m}
 = (\id_{R^m}\otimes\mathcal P_m)(\omega_{RB}^{\otimes m}).
 \label{egupper:eq:pinched-state}
\end{equation}
Then, for $0<s<1$ and $\alpha=1/(1-s)$,
\begin{align}
 &\sup_{\substack{\tau:\,\supp\tau\subseteq\supp\widehat\omega}}
       [-sH(R^m|B^m)_\tau-D(\tau\|\widehat\omega)]\nonumber\\
 &\qquad\ge-ms\widetilde H_\alpha^\uparrow(R|B)_\omega
                    -(1+s)\log v_m.
 \label{egupper:eq:uniform-pinching}
\end{align}
\end{lemma}

\begin{proof}
All conditioning states below are normalized, and logarithms of
$\widehat\omega$ are restricted to its support. Sion's theorem applied to
the compact state domain of $\tau$ and the full-rank conditioning states gives
\begin{align}
 &\sup_\tau[-sH(R^m|B^m)_\tau-D(\tau\|\widehat\omega)]
       \nonumber\\
 &=\inf_{\sigma_{B^m}>0}\sup_\tau
       [sD(\tau\|I_{R^m}\otimes\sigma_{B^m})
          -D(\tau\|\widehat\omega)].
 \label{egupper:eq:conditional-minimax}
\end{align}
The payoff is concave and continuous in $\tau$ and convex and continuous in
$\sigma_{B^m}$. The state $\widehat\omega$ is invariant under simultaneous
permutations of the $R$ and $B$ copies, since $\omega^{\otimes m}$ has this
symmetry and the pinching map commutes with permutations. Thus the supremum
over $\tau$ in Eq.~\eqref{egupper:eq:conditional-minimax} is unchanged by
permuting $\sigma_{B^m}$. By convexity, averaging $\sigma_{B^m}$ over these
permutations cannot increase this supremum, so the infimum may be restricted
to permutation-invariant states. For such states, universal domination and
operator monotonicity give
$\log\sigma_{B^m}\le\log\sigma^u_{B^m}+(\log v_m)I_{B^m}$.
Moreover, pinching ensures
$[\widehat\omega_{R^mB^m},I_{R^m}\otimes\sigma^u_{B^m}]=0$.
The Gibbs variational formula therefore evaluates the remaining supremum
as $s\widetilde D_\alpha(\widehat\omega\|
I_{R^m}\otimes\sigma^u_{B^m})$, giving
\begin{equation}
 \sup_\tau[-sH(R^m|B^m)_\tau-D(\tau\|\widehat\omega)]
 \ge s\widetilde D_\alpha(\widehat\omega\|
                    I_{R^m}\otimes\sigma^u_{B^m})-s\log v_m,
 \label{egupper:eq:universal-loss}
\end{equation}
as in Ref.~\cite[Proposition~5(iii)]{li2023strong}; see also
Ref.~\cite{mosonyi2017strong}. This is the first loss.

The pinching inequality $\omega^{\otimes m}\le v_m\widehat\omega$, followed
by congruence with
$(I_{R^m}\otimes\sigma^u_{B^m})^{-s/2}$ and eigenvalue monotonicity, gives
\begin{equation}
 s\widetilde D_\alpha(\widehat\omega\|
                    I_{R^m}\otimes\sigma^u_{B^m})
 \ge s\widetilde D_\alpha(\omega^{\otimes m}\|
                    I_{R^m}\otimes\sigma^u_{B^m})-\log v_m.
 \label{egupper:eq:pinching-loss}
\end{equation}
This is the second loss. Finally, minimizing the last divergence over its
conditioning state and using conditional-entropy additivity on products gives
the right side of Eq.~\eqref{egupper:eq:uniform-pinching}.
\end{proof}

\begin{proposition}[Entanglement-generation exponent]
\label{egupper:prop:bound}
For every finite-dimensional quantum channel $\mathcal N$ and $r\ge0$,
\begin{align}
 E_{\mathrm{sc,eg}}(r,\mathcal N)
 &\le\inf_{k\ge1}\sup_{\alpha>1}\frac{\alpha-1}{\alpha}
       \left[r-\frac1k\RenyiCoherentInformation{\alpha}(\mathcal N^{\otimes k})\right]
 \label{egupper:eq:finite-block-bound}\\
 &=\sup_{\alpha>1}\frac{\alpha-1}{\alpha}
       \big[r-\RenyiQuantumCapacity{\alpha}(\mathcal N)\big].
 \label{egupper:eq:bound}
\end{align}
Consequently, equality holds in \eqref{eq:exponent-bound}.
\end{proposition}

\begin{proof}
The case $r=0$ is immediate. To first obtain the unregularized bound, note the
exact coding identity
\begin{equation}
 F_{\mathrm{eg}}^\star(\ell,mr;\mathcal N^{\otimes m})
 =F_{\mathrm{eg}}^\star(m\ell,r;\mathcal N),
 \label{egupper:eq:block-fidelity}
\end{equation}
where the channel after the semicolon indicates which channel defines the
optimization. Taking logarithms and using the existence of the operational
limits established after Eqs.~\eqref{eq:exponent-definition} and
\eqref{holevo:eq:exponent-definition} therefore gives
$E_{\mathrm{sc,eg}}(mr,\mathcal N^{\otimes m})
=mE_{\mathrm{sc,eg}}(r,\mathcal N)$.

Apply Lemma~\ref{egupper:lem:convexification} to the block channel
$\mathcal N^{\otimes m}$ at rate $mr$. In its inner supremum, restrict
$\omega$ to the pinched states in Eq.~\eqref{egupper:eq:pinched-state}, with
$\omega_{RB}=(\id_R\otimes\mathcal N)(\psi_{RA})$ and $\psi$ pure. Such states
are produced by the block channel followed by Bob's local pinching.
Lemma~\ref{egupper:lem:pinching}, followed by the supremum over $\psi$, yields
\begin{align}
 mE_{\mathrm{sc,eg}}(r,\mathcal N)
 &\le \sup_{0<s<1}\Bigl\{ms
       [r-\RenyiCoherentInformation{1/(1-s)}(\mathcal N)]
       +(1+s)\log v_m\Bigr\}
 \label{egupper:eq:pinching-bound}\\
 &\le m\sup_{\alpha>1}\frac{\alpha-1}{\alpha}
       [r-\RenyiCoherentInformation{\alpha}(\mathcal N)]+2\log v_m.
 \label{egupper:eq:before-limit}
\end{align}
Since $v_m$ grows polynomially, division by $m$ and $m\to\infty$ remove the
remainder. Applying the resulting bound to $\mathcal N^{\otimes k}$ at rate
$kr$, and using the block identity once more, proves
Eq.~\eqref{egupper:eq:finite-block-bound}.

It remains to exchange regularization and the order supremum; compare
\cite[Lemma~28]{li2024operational}. Nonnegativity and monotonicity in the order
imply, for $0<s<t<1$,
\begin{align}
 t\left[r-\frac1k\RenyiCoherentInformation{1/(1-t)}(\mathcal N^{\otimes k})\right]
 &\le s\left[r-\frac1k\RenyiCoherentInformation{1/(1-s)}
                    (\mathcal N^{\otimes k})\right]+r(t-s).
 \label{egupper:eq:grid}
\end{align}
The dimension bound makes the expression tend to zero as $t\searrow0$.
Defining its value to be zero at $s=0$, a grid of mesh $\varepsilon$ therefore
approximates the supremum within $r\varepsilon$, uniformly in $k$. On a fixed
finite grid, Eq.~\eqref{eq:regularization} permits $k\to\infty$ to pass through
the maximum. Letting $\varepsilon\searrow0$ shows that the finite-block
suprema converge to the regularized supremum in
Eq.~\eqref{egupper:eq:bound}. 
Each finite-block supremum is at least this limit
by \eqref{eq:regularized-supremum}, so its infimum has the same value.
Combining this upper bound with Eq.~\eqref{eq:exponent-bound} completes the
proof.
\end{proof}

\section{Other quantum communication criteria}
\label{app:quantum-criteria}

Fix an $(n,M)$ transmission code with $M\ge2^{nr}$. Let $F_{\mathrm{et}}$ be the squared
entanglement-transmission fidelity, $F_{\mathrm{sub}}$ the minimum pure-state
fidelity without a reference, and $F_{\mathrm{ref}}$ the minimum fidelity over
pure inputs including an arbitrary reference. Write
$\Lambda=\mathcal D\circ\mathcal N^{\otimes n}\circ\mathcal E$ for the
effective logical channel. All fidelities use the squared convention. Every
transmission code is a generation code, and
\begin{equation}
 F_{\mathrm{ref}}\le F_{\mathrm{et}}
 \le F_{\mathrm{eg}}^{\star}(n,r).
 \label{eq:transmission-eg-comparison}
\end{equation}
The Haar-average identity
$F_{\mathrm{av}}=(MF_{\mathrm{et}}+1)/(M+1)$
\cite[Proposition~4.4]{kretschmann2004tema} gives
\begin{equation}
 F_{\mathrm{sub}}\le F_{\mathrm{av}}
 \le2F_{\mathrm{eg}}^{\star}(n,r),
 \label{eq:subspace-eg-comparison}
\end{equation}
where $F_{\mathrm{eg}}^{\star}(n,r)\ge1/M$ for a product preparation.
Thus exponential decay of the generation fidelity implies the same property
for all these criteria. It also implies exponential convergence of normalized
diamond-distance error to one, since
$1-\tfrac12\|\Lambda-\id_M\|_\diamond\le F_{\mathrm{et}}$.

The exact exponents need not coincide. For a replacer channel
$\mathcal N(\rho)=\sigma\Tr[\rho]$, the optimal generation and transmission
fidelities are $1/M$ and $1/M^2$, giving exponents $r$ and $2r$ at rate $r$.
The exact quantum exponent in Theorem~\ref{thm:operational} is therefore
specific to entanglement generation.

\section{Classical strong-converse exponent}
\label{app:classical-exponent}

\subsection{Converse bound and positivity}

For each $\alpha>1$, the one-shot converse yields
\begin{equation}
 -\frac1n\log P_{\mathrm s}^{\star}(n,r)
 \ge\frac{\alpha-1}{\alpha}
     \left[r-\frac1n\RenyiHolevoQuantity{\alpha}(\mathcal N^{\otimes n})\right].
\end{equation}
Taking $n\to\infty$ and then the supremum over $\alpha$ proves
the lower bound on $E_{\mathrm{sc,c}}(r,\mathcal N)$ in
Eq.~\eqref{holevo:eq:exponent-bound}. If $r>\ClassicalCapacity(\mathcal N)$,
choose $0<\varepsilon<r-\ClassicalCapacity(\mathcal N)$. By
Proposition~\ref{prop:finite-block-continuity}, one may fix $\alpha>1$ close
enough to one that
\begin{equation}
 \RenyiHolevoQuantity{\alpha}(\mathcal N^{\otimes n})
 \le n[\ClassicalCapacity(\mathcal N)+\varepsilon]\qquad(n\ge1).
\label{classicalupper:eq:uniform-continuity-consequence}
\end{equation}
Substitution into the one-shot inequality gives exponential decay with a
positive exponent. For a code sequence with limiting
inferior rate above capacity, choose $r$ strictly between them and
apply this bound for all sufficiently large $n$. This proves the classical assertions of
Theorem~\ref{thm:operational}.

\subsection{Matching upper bound from Mosonyi--Ogawa}
\label{app:classical-upper}

The matching upper bound for classical communication follows from
Mosonyi and Ogawa's coding bound for quantum channels
\cite[Theorem~VI.1]{mosonyi2017strong}. That theorem gives a
single-block upper bound and a regularized lower bound. We record only
the passage from its upper bound to the regularized equality needed here.
For the alphabet consisting of all states on the input of
$\mathcal N^{\otimes k}$, Proposition~IV.2 of that paper identifies their
optimized quantity $\chi_\alpha^*$ with
$\RenyiHolevoQuantity{\alpha}(\mathcal N^{\otimes k})$: finitely supported
input distributions are exactly the finite ensembles used in our definition
\cite[Proposition~IV.2 and Section~VI]{mosonyi2017strong}.

\begin{proposition}[Consequence of Mosonyi--Ogawa]
\label{classicalupper:prop:bound}
For every finite-dimensional quantum channel $\mathcal N$ and $r\ge0$,
\begin{align}
 E_{\mathrm{sc,c}}(r,\mathcal N)
 &=\inf_{k\ge1}\sup_{\alpha>1}\frac{\alpha-1}{\alpha}
       \left[r-\frac1k\RenyiHolevoQuantity{\alpha}(\mathcal N^{\otimes k})\right]
 \label{classicalupper:eq:finite-block-bound}\\
 &=\sup_{\alpha>1}\frac{\alpha-1}{\alpha}
       \big[r-\RenyiClassicalCapacity{\alpha}(\mathcal N)\big].
 \label{classicalupper:eq:bound}
\end{align}
\end{proposition}

\begin{proof}
The case $r=0$ is immediate. For $r>0$, apply the upper bound of
\cite[Theorem~VI.1]{mosonyi2017strong} to the block channel
$\mathcal N^{\otimes k}$ at rate $k(r+\varepsilon)$.
Use code sequences with arbitrarily small excess over its exponent bound.
They are admissible for $\mathcal N$ at lengths $km$;
the rate margin $\varepsilon>0$ ensures their rates are eventually at
least $r$. Evaluating the operational limit along these block lengths,
dividing by $k$, and letting $\varepsilon\searrow0$ gives
\begin{equation}
 E_{\mathrm{sc,c}}(r,\mathcal N)
 \le\sup_{\alpha>1}\frac{\alpha-1}{\alpha}
       \left[r-\frac1k\RenyiHolevoQuantity{\alpha}(\mathcal N^{\otimes k})\right]
 \qquad(k\ge1).
 \label{classicalupper:eq:fixed-block}
\end{equation}
The passage in rate is uniform in the order, since
$0<(\alpha-1)/\alpha<1$. This invokes their coding theorem directly;
no separate classical-quantum coding construction is required.

To pass to regularization, we use a finite grid argument; compare
Ref.~\cite[Lemma~28]{li2024operational}. Put
$f_k(s)=s[r-k^{-1}\RenyiHolevoQuantity{1/(1-s)}
(\mathcal N^{\otimes k})]$ for $0<s<1$, and $f_k(0)=0$.
The dimension bound gives $f_k(s)\to0$ as $s\searrow0$, so adjoining
$s=0$ does not change the supremum.
Nonnegativity and monotonicity of the R\'enyi information imply
\begin{equation}
 f_k(t)\le f_k(s)+r(t-s),\qquad 0\le s<t<1.
 \label{classicalupper:eq:grid}
\end{equation}
Consequently, a finite grid of mesh at most $\eta$, rounding each
$t<1$ down to a grid point, approximates $\sup_{0\le s<1}f_k(s)$
within $r\eta$, uniformly in $k$. Pointwise convergence from
\eqref{holevo:eq:definition} passes through the finite grid maximum.
Letting $k\to\infty$ and then $\eta\searrow0$ proves convergence of
the order suprema to the regularized expression in
\eqref{classicalupper:eq:bound}. Each finite-block supremum is at least
this limit by \eqref{holevo:eq:regularized-supremum}, so its infimum
also equals that limit. Combining \eqref{classicalupper:eq:fixed-block}
with the lower bound \eqref{holevo:eq:exponent-bound} proves both
equalities. The uniform classical order-one limit in
Theorem~\ref{thm:main} is not needed for this characterization;
it establishes its strict positivity for $r>\ClassicalCapacity(\mathcal N)$.
\end{proof}

\begin{proof}[Proof of Theorem~\ref{thm:operational}]
The product-code argument following Eqs.~\eqref{eq:exponent-definition}
and~\eqref{holevo:eq:exponent-definition} proves existence of both operational
limits. The converse-and-positivity arguments in
Appendices~\ref{app:quantum-exponent} and~\ref{app:classical-exponent}, together
with Propositions~\ref{egupper:prop:bound}
and~\ref{classicalupper:prop:bound}, prove the two equalities and their strict
positivity above the corresponding capacities.
\end{proof}

%

\end{document}